\documentclass[12pt,reqno]{amsart}
\usepackage{amsmath}%
\usepackage{amsthm}
\usepackage{amsfonts}%
\usepackage{amssymb}%
\usepackage{graphicx}
\usepackage{wrapfig, adjustbox}
\usepackage{caption}
\usepackage[foot]{amsaddr}
\usepackage{natbib}
\usepackage{apalike}

\usepackage{bbm}
\usepackage{enumerate}
\usepackage{xcolor}
\usepackage{subfig}
\usepackage[toc,page]{appendix}
\usepackage{apptools}
\usepackage[misc]{ifsym}
\usepackage{multirow}
\usepackage[margin=2.5cm]{geometry} % Package fuer die Einstellung der Raender
\newtheorem{theorem}{Theorem}
\theoremstyle{plain}

\newtheorem{definition}[theorem]{Definition}
\newtheorem{example}[theorem]{Example}

\newtheorem{lemma}[theorem]{Lemma}

\newtheorem{proposition}[theorem]{Proposition}

\newtheorem{assumption}[theorem]{Assumption}
\numberwithin{equation}{section}
\numberwithin{theorem}{section}
\AtAppendix{\counterwithin{lemma}{section}}

\newcommand{\Var}{\text{Var}}

\newcommand{\tr}{\text{Tr}}
\newcommand{\T}{\intercal}
\newcommand{\bb}{\backslash}

\newcommand{\Ikk}{\mathbb{R}^{k+1} \backslash \{0\}}
\newcommand{\I}{\mathbb{R} \backslash \{0\}}
\newcommand{\m}{\text{-}}
\DeclareMathOperator*{\argsup}{arg\,sup}

\begin{document}
\title[Mean-variance hedging of unit linked life insurance contracts in a L\'evy model]{Mean-variance hedging of unit linked life insurance contracts in a jump-diffusion model}
\author{Frank Bosserhoff*, Mitja Stadje*}
\address[]
{*Institute of Insurance Science and Institute of Financial Mathematics, Ulm University \newline%
\indent Helmholtzstrasse 20, 89081 Ulm, Germany}%
\email[]{\Letter\ frank.bosserhoff@uni-ulm.de, mitja.stadje@uni-ulm.de}%

\date{\today}
%\@namedef{subjclassname@2010}{%
 % \textup{2010} Mathematics Subject Classification}
\thanks{FB acknowledges support by \emph{Stiftung WiMa Ulm}.}
%\dedicatory{Dedicated to Professor XY on the occasion of his seventieth birthday.}

% Abstract Fudan
\begin{abstract} 
%We link the theory of mean-variance portfolio selection in a jump-diffusion model with the reformulation of a time-inconsistent optimal control problem in game-theoretic terms searching for a Nash subgame perfect equilibrium. Considering an insurer being able to trade in stocks and a longevity asset under the incorporation of basis risk, we characterize the dynamically mean-variance optimal trading strategy as the solution of a system of PIDEs, a so called extended HJB system. We prove that the equilibrium is necessarily a solution of the extended HJB system. Under certain conditions we obtain an explicit solution to the extended HJB system, thereby providing the optimal trading strategies in closed-form. 
We consider a time-consistent mean-variance portfolio selection problem of an insurer and allow for the incorporation of basis (mortality) risk. The optimal solution is identified with a Nash subgame perfect equilibrium. We characterize an optimal strategy as solution of a system of partial integro-differential equations (PIDEs), a so called extended Hamilton-Jacobi-Bellman (HJB) system. We prove that the equilibrium is necessarily a solution of the extended HJB system. Under certain conditions we obtain an explicit solution to the extended HJB system and provide the optimal trading strategies in closed-form. A simulation shows that the previously found strategies yield payoffs whose expectations and variances are robust regarding the distribution of jump sizes of the stock. The same phenomenon is observed when the variance is correctly estimated, but erroneously ascribed to the diffusion components solely. Further, we show that differences in the insurance horizon and the time to maturity of a longevity asset do not add to the variance of the terminal wealth.
\end{abstract}
\maketitle

%%%%%%%%%%%%%%%%%%%%%%%%%%%%%%%%%%%%%%%%%%%%%%%%%%%%%%%%%%%%%%%%%%%%%%%%%%%%%%%%%
\section{Introduction}
%%%%%%%%%%%%%%%%%%%%%%%%%%%%%%%%%%%%%%
Two major risks faced by life insurance companies are longevity and asset risk. Longevity risk refers to the risk that the future changes in the mortality rates are incorrectly estimated while asset risk refers to the possibility of a future loss in the investment portfolio. Increases in the life expectancy might among others stem from sudden changes in environmental or medical conditions that are not foreseeable upon the contract initiation. Clearly, these changes need to be accounted for by the mortality model. An adequate way to do so is the modeling of the force of mortality with a diffusion process supplemented by jumps, see \cite{vigna} and \cite{blake} for a detailed discussion. Another practical challenge of hedging longevity risk is basis risk. The payoff of for instance a longevity bond depends on a particular mortality rate that is related to but certainly not identical with the insurer's portfolio, see \cite{br1} and \cite{br2} for empirical studies on this issue. Thus, buying a longevity asset can only provide a partial hedge against an insurance company's mortality exposure. \\
Asset risk stems from the fact that the premiums paid by the insured are to be gainfully invested at the capital market inducing financial risk. Empirical evidence suggests that returns are non-gaussian and leptokurtic, see e.g. \cite{ct}, \cite{schoutens} and references therein. Thus, in order to properly capture this risk, an insurer should base a stock price model on a Brownian component and additionally allow for jumps. Such a financial market is known to be generically incomplete. Consequently, an insurance company facing the aforementioned risks cannot perfectly hedge its obligations. Hence, a way to quantify risk needs to be specified. In this paper we identify risk with the variance of the terminal wealth. The identification of risk with the variance of the terminal payoff has a long tradition in academia as well as industry and dates back already to \cite{markowitz}. Mean-variance portfolio selection is intuitively appealing and analytically tractable. A major drawback, however, is \emph{time-inconsistency}, which means that due to the non-linearity and non-recursiveness of the variance part, the dynamic programming approach fails. An investor might initiate a dynamic strategy because it is optimal at a particular point in time, knowing fully well that she will deviate from this strategy later on. Investors ignoring the sub-optimality of a previously found strategy are said to \emph{pre-commit}, see \cite{strotz} for fundamentals on this problem. In \cite{xyz} and \cite{lim}, the pre-commitment version of a mean-variance portfolio selection problem in a continuous-time economy is solved. The question of dynamically optimal, i.e., time-consistent mean-variance policies has been addressed for example by \cite{basak}. Their solution approach is based on a recursive formulation allowing for the application of dynamic programming. The authors point out that the same solution could be found as the Nash subgame perfect equilibrium outcome whereby the investor is playing a game with a future incarnation of herself, that is, a game with infinitely many players. Thereby a strategy is a Nash subgame perfect equilibrium if at some given point in time an investor knows that every future "player" will follow a certain strategy, then it is optimal for her not to deviate. For a general game-theoretic background we refer to \cite{hp} and references therein. The reformulation of time-inconsistent control problems in game-theoretic terms has been originally proposed by \cite{ekeland} and \cite{bm}. This line of research has been followed by \cite{basak}, \cite{wang}, \cite{czichowsky}, \cite{bensoussan} and \cite{lindensjo}. \\
It is well known that the optimal value function of a standard time-consistent stochastic control problem can be characterized as the unique solution of a non-linear partial integro-differential equation (PIDE), see \cite{oksjump}, known as Hamilton-Jacobi-Bellman (HJB) equation.  In \cite{bm} it is shown that the reformulation of a time-inconsistent control problem in game-theoretic terms leads to a system of nonlinear PIDEs, the so called \emph{extended HJB system}. Further, they provide a verification result showing that solving the extended HJB system is a sufficient condition for being an equilibrium control law. \cite{lindensjo} proves that under certain regularity assumptions solving the extended HJB system is a necessary condition for being an equilibrium; however, the proof is restricted to the diffusion case. \\
We consider a continuous-time Markovian economy in which an insurer trades in an arbitrary quantity of risky financial assets, a zero-coupon longevity bond and a riskless asset in order to hedge some terminal payout with regard to mean-variance optimality. Thereby the underlying financial assets as well as the force of mortality are modeled by jump-diffusions. Our first main contribution is an extension of the work of \cite{lindensjo} by proving that an equilibrium necessarily solves the extended HJB system. Secondly, for the case that an insurer neglects the hedge of some terminal payoff, we are able to present explicit closed-form solutions for the optimal trading strategies, the equilibrium value function and the expected terminal wealth. Thirdly, we exemplify our findings along a tractable model and provide numerical as well as graphical illustrations. When the jumps are erroneously neglected while the expected values and variances of the stock price and the force of mortality are correctly determined, the expected optimal terminal payoff and its variance hardly change. Thus, the time-consistent mean-variance optimal terminal wealth is robust regarding the consideration of jumps. A similar result is found when changing the distribution of the jump sizes of the stock. As the market for longevity assets is relatively illiquid, one can certainly not expect to find a hedging instrument whose time to maturity coincides with the insurance horizon. We find that this does not effect the expectation and variance of the optimal final payoff in our setup. Moreover, we find that our strategies significantly outperform the gain from investing in the riskless asset only.   \\
The remainder of this paper is structured as follows: In Section 2 we introduce the financial and longevity market under consideration and clarify what is meant by an admissible trading strategy. In the following section we turn to the optimization problem by gradually defining several auxiliary functions, operators and the notion of equilibrium control. With these at hand, Section 3 closes with the specification of the extended HJB system. In Section 4 our first main result (Theorem \ref{thm_nec}) is stated and proved. In Section 5 we neglect the hedge of a terminal payout and present an explicit equilibrium solution, the main result here is Theorem \ref{thm_strategies}. The final section contains the numerical applications. \\

\textbf{Notation.} Denote by $\mathbb{R}_{+}$ the positive real numbers and by $\mathbb{R}^{m}_{+}$ its $m$-fold Cartesian product. The zero vector in any Euclidean space $\mathbb{R}^{m}$ is written as $0$. For $x \in \mathbb{R}^{m}$ we use $||x||_{2}:= \sqrt{\sum_{i=1}^{m} |x_{i}|^{2}}$. The symbol $\mathbb{R}^{m \times n}$ denotes the space of real-valued matrices with $m$ rows and $n$ columns. If $x \in \mathbb{R}^{m}$, the matrix Diag($x$) $\in \mathbb{R}^{m \times m}$ is the square matrix with the entries of $x$ on the diagonal and all off-diagonal elements being equal to zero. For $A \in \mathbb{R}^{m \times n}$, the symbol $A^{\T}$ means the transposed of $A$. If $A$ is a square matrix, we write $\tr(A)$ for its trace. For any $T > 0,\ t \in (0,T)$ and some function $f: [0,T] \times \mathbb{R}^{m} \to \mathbb{R}$ such that $f \in C^{1,2}$, we define $\dot{f}(t,\cdot) := \frac{\partial}{\partial t}f(t,\cdot)$ and for any $x \in \mathbb{R}^{m}$ we denote for arbitrary $j \in \{1,\dots,m\}$ by $f_{x_{j}}(\cdot,x)$ the first order partial derivative of $f$ w.r.t. the $j$th component of $x$. Moreover, let the gradient of $f$ w.r.t. $x$ be denoted by $\nabla_{x}f(\cdot,x)$ and the Hessian matrix by $H_{x}f(\cdot,x).$

%%%%%%%%%%%%%%%%%%%%%%%%%%%%%%%%%%%%%%%%%%%%%%%%%%%%%%%%%%%%%%%%%%%%%%%%%%%%%%%%%
% \clearpage
\section{Model Setup} 
%%%%%%%%%%%%%%%%%%%%%%%%%%%%%%%%%%%%%%
Let $T > 0$ be the planning horizon. Consider the probability space $(\Omega, \mathcal{F}, \mathbb{P})$ that is equipped with a standard $d+1$-dimensional Brownian motion $\hat{W}:=(W^{1}, \dots, W^{d},W^{d+1})^{\T},$ whereby we define $W:=(W^{1},\dots,W^{d})^{\intercal}$ and $\bar{W}:=W^{d+1},$ and a Poisson random measure $J_{\hat{X}}(dt,d\hat{x})$ on $[0,T] \times \mathbb{R}^{k+1} \backslash \{0\},$ independent of $\hat{W},$ with respective intensity measure $\vartheta_{\hat{X}}(d\hat{x})dt.$ Denote its compensated version by $\tilde{J}_{\hat{X}}(dt,d\hat{x}) = J_{\hat{X}}(dt,d\hat{x}) - \vartheta_{\hat{X}}(d\hat{x})dt.$ Let $(\mathcal{F}_{t})_{t \in [0,T]}$ be the right-continuous completion of the filtration generated by $\hat{W}$ and $J_{\hat{X}}$. Throughout this paper we impose the following condition:

\begin{assumption}
The L\'evy measure $\vartheta_{\hat{X}}$ is such that
\begin{equation}
\int_{\Ikk} |\hat{x}|^{2}\ \vartheta_{\hat{X}}(d\hat{x}) < \infty.
\label{eq:secmom_fin}
\end{equation}
\end{assumption}

Under \eqref{eq:secmom_fin}, let $\hat{X} := (X^{1},\dots,X^{k},X^{k+1})^{\T}$ be a vector of pure-jump independent $(\mathcal{F}_{t})$-martingales with $X_{t}^{j} = \int_{0}^{t} \int_{\Ikk}\hat{x}^{j}\ \tilde{J}_{\hat{X}}(ds,d\hat{x})$, $j=1,\dots,k,k+1,$ whereby $\hat{x}^{j}$ is the $j$th coordinate of $\hat{x} \in \mathbb{R}^{k+1}.$ We define $X:=(X^{1},\dots,X^{k})^{\T}$ and $\bar{X}:= X^{k+1},$ so $\hat{X}=(X,\bar{X})^{\T}.$ Further, we may write $J_{\hat{X}}(dt,d\hat{x}) = J_{X,\hat{X}}(dt,dx,d\bar{x}) = \mathbbm{1}_{\bar{x} = 0}\ J_{X}(dt,dx) + \mathbbm{1}_{x=0}\ J_{\bar{X}}(dt,d\bar{x}).$ For any $E \subseteq \mathbb{R}^{k+1},$ the independence of $X$ and $\bar{X}$ implies that (cf. \cite{ct}, Proposition 5.3)
$$\vartheta_{X,\bar{X}}(E) = \vartheta_{X}(E_{X}) + \vartheta_{\bar{X}}(E_{\bar{X}}),$$
with
\begin{align*}
E_{X} &:= \{x \in \mathbb{R}^{k}: (x,0) \in E\}, \\
E_{\bar{X}} &:= \{\bar{x} \in \mathbb{R}: (0,\bar{x}) \in E\}.
\end{align*}

Assume that $x^{j} > -1$ for all $j \in \{1,\cdots,k\}.$ The financial market under consideration consists of a bank-account paying interest at a deterministic rate $r \geq 0$ and $m$ risky stocks, $1 \leq m \leq \min \{d,k\},$ with price processes $S^{i} = (S_{t}^{i})_{t \in [0,T]}$, $i = 1,\dots,m,$ satisfying the SDEs given by

\begin{equation}
\frac{dS_{t}^{i}}{S_{t\m}^{i}} = \mu_{i} \ dt + \sum_{j=1}^{d} \sigma_{ij}\ dW_{t}^{j} + \sum_{j=1}^{k} \rho_{ij} \ dX_{t}^{j}, 
\label{eq:S}
\end{equation}

where $S_{0}^{i} = s_{i} \in \mathbb{R}_{+},\ \mu_{i} \in \mathbb{R},\ \sigma_{ij} \in \mathbb{R}_{+}$ and $\rho_{ij} \in \mathbb{R}_{+}$ such that $\sum_{j=1}^{k} \rho_{ij} \leq 1$ respectively denote the initial price of stock $i$, the rate of appreciation, the volatilities and the jump-sensitivities. We assume that the financial market is free of arbitrage, i.e., there exists a measure $\mathbb{Q}$ that is equivalent to $\mathbb{P}$ such that the discounted stock price processes $(S_{t}^{i} / e^{rt})_{t \in [0,T]}, i=1,\dots,m,$ are $(\mathcal{F}_{t})$- martingales under $\mathbb{Q}$.\\
In addition to the financial market, we consider an arbitrage-free mortality market on which an investor can buy a zero-coupon longevity bond. 
We use the Brownian motion $\bar{W}$ and the jump component $\bar{X}$ to model the force of mortality. In particular, the force of mortality $\lambda$ shall be given as the solution of the SDE

\begin{equation}
d\lambda_{t} = \mu_{\lambda}(t,\lambda_{t}) \ dt + \sigma_{\lambda}(t,\lambda_{t}) \ d\bar{W}_{t} + \int_{\I} \tilde{\sigma}_{\lambda}(t,\lambda_{t\m},\bar{x}) \ \tilde{J}_{\bar{X}}(dt,d\bar{x}),
\label{eq:fom}
\end{equation}
whereby $\lambda_{0} \in \mathbb{R}_{+},$ and $\mu_{\lambda}, \sigma_{\lambda}$ and $\tilde{\sigma}_{\lambda}$ satisfy Assumption \ref{ass_func_lambda} below. We remark that the force of mortality can become negative with positive probability. In practical applications it is therefore common to chose $\mu_{\lambda}$ high and $\sigma_{\lambda}$ as well as $\tilde{\sigma}_{\lambda}$ small enough, see \cite{vigna} for a discussion on calibration. We take $\lambda_{t} > 0$ for all $t \in [0,T].$

\begin{assumption}
We assume that $\mu_{\lambda}, \sigma_{\lambda}: [0,T] \times \mathbb{R}_{+} \to \mathbb{R}$ and $\tilde{\sigma}_{\lambda}:[0,T] \times \mathbb{R}_{+} \times \mathbb{R} \backslash \{0\} \to \mathbb{R}$ satisfy the following conditions:\\

\begin{enumerate}[(i)]
	\item (At most linear growth) There exists a constant $B_{1} < \infty$ such that for all $a \in \mathbb{R}_{+}$ it holds that
		\begin{align*}
			|\mu_{\lambda}(t,a)|^{2} + |\sigma_{\lambda}(t,a)|^{2} + |\tilde{\sigma}_{\lambda}(t,a,\bar{x})|^{2}  &\leq B_{1}(1+|a|^{2}).
				\end{align*}
	
	\item 
 (Uniform Lipschitz continuity) There exists a constant $C_{1} < \infty$ such that for all $a,b \in \mathbb{R}_{+}$ it holds that
\begin{align*}
&|\mu_{\lambda}(t,a) - \mu_{\lambda}(t,b)|^{2} + |\sigma_{\lambda}(t,a) - \sigma_{\lambda}(t,b)|^{2}\\
& \ + \int_{\mathbb{R} \backslash \{0\}}|\tilde{\sigma}_{\lambda}(t,a,\bar{x}) - \tilde{\sigma}_{\lambda}(t,b,\bar{x})|^{2}\ \vartheta_{\bar{X}}(d\bar{x}) \leq C_{1} |b-a|^{2}.
\end{align*}
\end{enumerate}
\label{ass_func_lambda}
\end{assumption}

We consider a longevity bond where the reference cohort is assumed to satisfy the following:
\begin{itemize}
	\item at time $t=0$, all members of the cohort are of the same age,
	\item the force of mortality of the cohort is entirely described by $\lambda$, 
	\item the cohort is sufficiently large such that the idiosyncratic risk is pooled away.
\end{itemize}

In addition, we assume that the insurance's planning horizon $T$ and the time to maturity of the longevity bond coincide. An investor who has bought the zero-coupon longevity bond at time $0 \leq t_{1} \leq T$ paying $L_{\lambda}(t_{1},T)$ receives $\exp\left(- \int_{t_{1}}^{T} \lambda_{s}\ ds \right)$ at time $T$. Consequently, the price $L_{\lambda}(t_{1},T)$ is given by
\begin{equation}
L_{\lambda}(t_{1},T) = \mathbb{E}_{\mathbb{Q}} \left[e^{-\int_{t_{1}}^{T}(\lambda_{s}+r) \ ds} \big| \mathcal{F}_{t_{1}} \right].
\label{eq:L_lambda}
\end{equation}
 Let $0 \leq t_{1} < t_{2} \leq T,$ suppose investor $A$ has bought the longevity bond at time $t_{1}$ at price $L_{\lambda}(t_{1},T)$ and there is a second investor, say $B$, who has bought the bond at time point $t_{2}$ paying $L_{\lambda}(t_{2},T).$ As the final payoff depends on the length of the holding period, it is clear that investor $A$ would not have sold her bond to $B$ at price $L_{\lambda}(t_{2},T)$ at time $t_{2}$, but she would have demanded a price of $\exp\left(-\int_{t_{1}}^{t_{2}} \lambda_{s}\ ds \right) L_{\lambda}(t_{2},T).$ Therefore, if an investor has bought the longevity asset at time $t_{1}$, the dollar value of her investment at any time $t_{2} > t_{1}$ is given by $Y_{t_{2}} := \exp\left(-\int_{t_{1}}^{t_{2}} \lambda_{s}\ ds \right) L_{\lambda}(t_{2},T).$ We name the dollar value process $Y$ from now on; the discounted version of $Y$ should be a martingale under the same risk-neutral measure $\mathbb{Q}.$ 

\begin{assumption}
We assume that $(\lambda, Y)$ is a Markovian It\^{o} jump-diffusion satisfying 
\begin{equation}
\frac{dY_{t}}{Y_{t-}} = (r + \nu_{L}(t,\lambda_{t},Y_{t})) \ dt + \sigma_{L}(t,\lambda_{t},Y_{t}) \ d\bar{W}_{t} + \int_{\mathbb{R} \setminus \{0\}} \eta_{L}(t,\lambda_{t\m},Y_{t\m},\bar{x}) \ \tilde{J}_{\bar{X}}(dt,d\bar{x}),
\label{eq:Y}
\end{equation}
with $Y_{0} = L_{\lambda}(0,T)$ and deterministic functions $\nu_{L}, \sigma_{L}, \eta_{L}.$ 
\label{ass_markov}
\end{assumption}

%\begin{proposition}
%Consider the process $\lambda$ given by \eqref{eq:fom}, suppose Assumption \ref{ass_func_lambda} and Assumption \ref{ass_markov} are satisfied. Then the process $Y$ solves the SDE
%\begin{equation}
%\frac{dY_{t}}{Y_{t-}} = (r + \nu_{L}(t,\lambda_{t},Y_{t})) \ dt + \sigma_{L}(t,\lambda_{t},Y_{t}) \ d\bar{W}_{t} + \int_{\mathbb{R} \setminus \{0\}} \eta_{L}(t,\lambda_{t\m},Y_{t\m},\bar{x}) \ \tilde{J}_{\bar{X}}(dt,d\bar{x}),
%\label{eq:Y}
%\end{equation}
%with $Y_{0} = L_{\lambda}(0,T)$ and deterministic functions $\nu_{L}, \sigma_{L}, \eta_{L}.$ 
%\end{proposition}
%\begin{proof}
%Since the discounted version of $Y$ is a martingale w.r.t. the pricing measure $\mathbb{Q}$, it holds by martingale representation that
%\begin{equation}
%\frac{dY_{t}}{Y_{t\m}} = r \ dt + \sigma_{L}(t,\lambda_{t},Y_{t}) \ d\bar{W}^{\mathbb{Q}}_{t} + \int_{\mathbb{R} \setminus \{0\}} \eta_{L}(t,\lambda_{t\m},Y_{t\m},\bar{x}) \ \tilde{J}^{\mathbb{Q}}_{\bar{X}}(dt,d\bar{x}),
%\label{eq:Y_Q}
%\end{equation}
%with $\bar{W}^{\mathbb{Q}}$ being a $\mathbb{Q}$-Brownian motion and $\tilde{J}^{\mathbb{Q}}_{\bar{X}}$ is the $\mathbb{Q}$-compensated version of $J_{\bar{X}}(dt,d\bar{x}).$ Note that Assumption \ref{ass_markov} implies that $\sigma_{L}$ and $\eta_{L}$ must be deterministic functions. An application of Girsanov's theorem then transforms \eqref{eq:Y_Q} into \eqref{eq:Y}. 
%\end{proof}

Note that the function $\nu_{L}$ in \eqref{eq:Y} is the market price of longevity risk. We further need the following assumption:

\begin{assumption}
We assume that $\nu_{L}, \sigma_{L}: [0,T] \times \mathbb{R}_{+} \times \mathbb{R}_{+} \to \mathbb{R}$ and $\eta_{L}: [0,T] \times \mathbb{R}_{+} \times \mathbb{R}_{+} \times \mathbb{R} \bb \{0\} \to \mathbb{R}$ satisfy the following conditions: 
\begin{enumerate}[(i)]
	\item (At most linear growth) There exists a constant $B_{2} < \infty$ such that for all $a,b \in \mathbb{R}_{+}$ it holds that
		\begin{align*}
			|\nu_{L}(t,a,b)|^{2} + |\sigma_{L}(t,a,b)|^{2} + \int_{\mathbb{R} \backslash \{0\}} |\eta_{L}(t,a,b,\bar{x})|^{2} \ \vartheta_{\bar{X}}(d\bar{x}) &\leq B_{2}(1+|a|^{2}+|b|^{2}).			
		\end{align*}
	\item 
 (Uniform Lipschitz continuity) There exists a constant $C_{2} < \infty$ such that for all $a_{1},a_{2},b_{1},b_{2} \in \mathbb{R}_{+}$ it holds that
	\begin{align*}
&|b_{1}\nu_{L}(t,a_{1},b_{1}) - b_{2}\nu_{L}(t,a_{2},b_{2})| + |b_{1}\sigma_{L}(t,a_{1},b_{1}) - b_{2}\sigma_{L}(t,a_{2},b_{2})|\\
&\ + \int_{\mathbb{R} \backslash \{0\}}|b_{1}\eta_{L}(t,a_{1},b_{1},\bar{x}) - b_{2}\eta_{L}(t,a_{2},b_{2},\bar{x})|^{2}\ \vartheta_{\bar{X}}(d\bar{x}) \leq C_{2}\ ||(a_{1},b_{1})-(a_{2},b_{2})||_{2}.
\end{align*}
\end{enumerate}
\label{ass_functions_Y}
\end{assumption}

We now consider an insurer who can invest in the $m$ risky stocks, deposit money in the bank account and use the longevity asset to partially hedge against its mortality exposure.  
Let $U \subseteq \mathbb{R}^{m+1}.$ An allocation rule is a predictable function $u:[0,T] \to U,\ t \mapsto (u_{S}(t),u_{Y}(t))^{\intercal},$ 
whereby $u_{S} = (u_{S^{1}},\dots,u_{S^{m}})^{\T}$ denotes the dynamic allocation process that indicates the total wealth that is invested in the stocks $1,\dots,m,$ and $u_{Y}$ the total wealth invested in the longevity asset. The portfolio process of the insurance company using the allocation rule $u$ is denoted by $P^{u} = (P^{u}_{t})_{t \in [0,T]}$ and fulfills the SDE

\begin{equation}
dP_{t}^{u} = u_{S}^{\intercal}(t\m)\ \frac{dS_{t}}{S_{t\m}} + u_{Y}(t\m) \ \frac{dY_{t}}{Y_{t\m}} + (P_{t}^{u} - u_{S}^{\intercal}(t) \textbf{1}- u_{Y}(t))r \ dt,
\label{eq:sde_portf_proc}
\end{equation}

with initial wealth $P_{0}^{u} = p > 0$ and \textbf{1}$ \in \mathbb{R}^{m}$ denotes a column vector of ones. Observe that $P^{u}$ as defined in \eqref{eq:sde_portf_proc} is self-financing.
\begin{definition}
An allocation rule $u$ is \emph{admissible} if for any point $(t,p) \in [0,T) \times \mathbb{R}_{+}$ there exists a unique c\`{a}dl\`{a}g adapted solution $P^{u}$ to \eqref{eq:sde_portf_proc} such that $\mathbb{E}[|P_{t}^{u}|^{2}] < \infty$ for all $t$. We denote by $\mathcal{U}$ the set of admissible allocation rules.
\end{definition}

%%%%%%%%%%%%%%%%%%%%%%%%%%%%%%%%%%%%%%

\section{Optimization Problem}
Classical mean-variance portfolio selection aims at finding a strategy that simultaneously maximizes the expected terminal payoff of a portfolio while minimizing its variance. We first consider the more general case where an insurance company trades in the financial and longevity market in order to hedge a terminal condition. Before rigorously defining what is meant by an equilibrium control in a stochastic optimization problem, we need some more notation and a target functional. Let $Z := (S^{1},\dots,S^{m},\lambda,Y)\in \mathbb{R}^{m+2}_{+}$ be the vector containing the traded assets as well as the force of mortality $\lambda$. Let $H = (H_{t})_{t \in [0,T]}$ be an $l$-dimensional Markovian jump-diffusion adapted to $(\mathcal{F}_{t})$ and $D: \mathbb{R}^{l} \to \mathbb{R}$ some function. The goal is a mean-variance optimal hedge of $D(H_{T})$ using $P^{u}$. %We need the following condition:

%\begin{assumption}
%There exists a sequence of processes $(H_{n})$ such that
%\begin{itemize}
	%\item $L^{2} \text{-} \lim_{n \to \infty} D(H_{n,t})= D(H_{t}),$ uniformly for $t \leq T$,
	%\item for any point $(t,h) \in [0,T) \times \mathbb{R}^{l}$ and some positive sequence $(c_{k})_{k \in \mathbb{N}}$ satisfying \\ $\lim_{k \to \infty} c_{k} = 0,$ it holds for the stopping time $\tau_{k} := \inf \{s>t: (s,H_{n,s}) \notin [t,t+c_{k}) \times B_{H}\}$ that $\tau_{k} > t$ and $\lim_{k \to \infty} \tau_{k} = t.$ Hereby $B_{H}$ is some arbitrary ball centered at $h$.
%\end{itemize}
%\label{ass_terminal}
%\end{assumption}

\begin{example}
Consider a process $\hat{\lambda} = (\hat{\lambda}_{t})_{t \in [0,T]}$ solving the SDE 
\begin{equation*}
d\hat{\lambda}_{t} = \mu_{\hat{\lambda}}(t,\hat{\lambda}_{t}) \ dt + \sigma_{\hat{\lambda}}(t,\hat{\lambda}_{t}) \ d\bar{W}_{t} + \int_{\I} \tilde{\sigma}_{\hat{\lambda}}(t,\hat{\lambda}_{t\m},\bar{x}) \ \tilde{J}_{\bar{X}}(dt,d\bar{x}),
\end{equation*}
with $\hat{\lambda}_{0} > 0.$ Suppose $\hat{\lambda}$ describes the force of mortality of the pool of insured persons, let $m=1$ for ease of exposition. If the insurance needs to deliver one share of $S$ to each person in its pool that has survived until the terminal time $T$, then the obligation $D(H_{T}) = S_{T}\ e^{-\int_{0}^{T} \hat{\lambda}_{s} \ ds}$ is to be hedged. 
\label{ex_H}
\end{example}

We write $\mathbb{E}_{t,p,z,h}[\cdot] = \mathbb{E}[\cdot | P_{t} = p, Z_{t} = z, H_{t} = h]$ for the conditional expectation given $(t,p,z,h) \in [0,T) \times \mathbb{R}_{+} \times \mathbb{R}^{m+2}_{+} \times \mathbb{R}^{l}$ and $\Var_{t,p,z,h}$ denotes the conditional variance accordingly. Let $\gamma > 0$ be a risk-aversion parameter.

\begin{definition}
For each $u \in \mathcal{U}$ and $\gamma > 0,$ we define the functions $F_{u}, g_{u}: [0,T] \times \mathbb{R}_{+} \times \mathbb{R}^{m+2}_{+} \times \mathbb{R}^{l} \to \mathbb{R}$ by
\begin{align}
\begin{split}
g_{u}(t,p,z,h) &= \mathbb{E}_{t,p,z,h}[P_{T}^{u}-D(H_{T})], \\
F_{u}(t,p,z,h) &= \mathbb{E}_{t,p,z,h}\left[P_{T}^{u} - \frac{\gamma}{2} (P_{T}^{u})^{2} + \gamma P_{T}^{u} D(H_{T}) - D(H_{T}) - \frac{\gamma}{2} D(H_{T})^{2}\right].
\label{eq:aux_fct}
\end{split}
\end{align}
\label{aux_fct}
\end{definition}

We also need to define the following differential operator:

\begin{definition}
To any vector $u \in \mathcal{U}$ we associate the operator $\mathcal{A}^{u}: f \mapsto \mathcal{A}^{u}f$ mapping $C^{1,2,2,2}([0,T] \times \mathbb{R}^{+} \times \mathbb{R}^{m}_{+} \times \mathbb{R}^{l}, \mathbb{R})$ to $C^{0,0,0,0}([0,T] \times \mathbb{R}^{+} \times \mathbb{R}^{m}_{+} \times \mathbb{R}^{l}, \mathbb{R})$ given by 
\begin{equation}
\mathcal{A}^{u}f(t,p,z,h) = \lim_{\epsilon \downarrow 0} \frac{\mathbb{E}_{t,p,z,h}[f(t+\epsilon,P_{t+\epsilon}^{u},Z_{t+\epsilon},H_{t+\epsilon})] - f(t,p,z,h)}{\epsilon}\ \text{(if the limit exists)}.
\label{eq:def_gen}
\end{equation}
\label{def_gen}
\end{definition}

The differential operator introduced in Definition \ref{def_gen} is known as \emph{infinitesimal generator of the graph} of the process $(P^{u},Z,H).$
Two further differential operators are needed; we presume them to act on suitably differentiable functions $f:$

\begin{itemize}
	\item $\mathcal{L}^{u}f(t,p,z,h) := \mathcal{A}^{u}f(t,p,z,h) - \dot{f}(t,p,z,h),$ which is called \emph{infinitesimal generator} of the process $(P^{u},Z,H)$,
	\item $\mathcal{G}^{u}f(t,p,z,h) := \gamma f(t,p,z,h) \mathcal{L}^{u} f(t,p,z,h) - \frac{\gamma}{2} \mathcal{L}^{u} f^{2}(t,p,z,h).$
\end{itemize}

\begin{definition} We define the\emph{ value function} $J:[0,T]\times \mathbb{R}_{+} \times \mathbb{R}^{m+2}_{+} \times \mathbb{R}^{l} \times \mathcal{U} \to \mathbb{R} $ by
\begin{align*}
J(t,p,z,h,u) &:= \mathbb{E}_{t,p,z,h}[P_{T}^{u} - D(H_{T})] - \frac{\gamma}{2}\ \emph{Var}_{t,p,z,h}[P_{T}^{u} - D(H_{T})] \\
&= F_{u}(t,p,z,h) + \frac{\gamma}{2}\ g_{u}^{2}(t,p,z,h).
\end{align*}
\label{def_value_fct}
\end{definition}

The second equality in Definition \ref{def_value_fct} easily follows from \eqref{eq:aux_fct}.
Finding some $u^{\star} \in \mathcal{U}$ such that $J(t,p,z,h,u)$ is \emph{maximal} is a time-inconsistent control problem and induces a path an investor would not follow. Therefore we next introduce the concept of equilibrium control. 

\begin{definition}\
\begin{itemize}
\item A trading strategy $u^{\star} \in \mathcal{U}$ is an \emph{equilibrium control} if 
\begin{equation}
 \liminf_{c \to 0} \frac{J(t,p,z,h,u^{\star}) - J(t,p,z,h,u_{t+c})}{c} \geq 0,
\label{eq:equ_control}
\end{equation}
for any  $(t,p,z,h) \in [0,T) \times \mathbb{R}_{+} \times \mathbb{R}^{m+2}_{+} \times \mathbb{R}^{l}$ and for all
$$u_{t+c} :=
\begin{cases}
u, \ \ on\ [t,t+c] \times B_{p} \times B_{z} \times B_{h}, \\
u^{\star}, \ \ on\  \{[t,t+c] \times B_{p} \times B_{z} \times B_{h}\}^{c},
\end{cases}$$
$t+c \leq T,$ where $u \in \mathcal{U}$ and $B_{p},B_{z},B_{h}$ are some arbitrary balls centered at respectively $p,z,h$.\\
\item The \emph{equilibrium value function} is defined by 
$$V(t,p,z,h) := J(t,p,z,h,u^{\star}).$$
\item An equilibrium policy $u^{\star}$ is of \emph{feedback type} if, for some \emph{feedback function}\\ $u_{\star}: [0,T] \times \mathbb{R}_{+}\times \mathbb{R}^{m+2}_{+} \times \mathbb{R}^{l} \to \mathcal{U},$ we have
\begin{equation}
u_{t}^{\star} = u_{\star}(t,P_{t-}^{\star},Z_{t-},H_{t-}), \ t \in [0,T],
\label{eq:feedback}
\end{equation}
with $P_{0-}^{\star} = p$, $Z_{0-} = Z_{0}$ and $H_{0-} = H_{0}.$ 
\label{def_equ}
\end{itemize}
%\label{def:eqcontrol}
\end{definition}

We see from \eqref{eq:equ_control} that a strategy is an equilibrium if a deviation is suboptimal given the knowledge that every future player will obey that strategy. In the sequel we will search for an equilibrium control law of feedback type. Recall that the optimal value function of a standard time-consistent stochastic optimal control problem is the solution of a partial integro-differential equation (PIDE) known as Hamilton-Jacobi-Bellman (HJB) equation. In \cite{bm} a similar approach for time-inconsistent stochastic optimal control problems is introduced leading to a system of PIDEs. The system to be solved is subsequently referred to as \emph{extended HJB system} and reduces to the classical case for a time-consistent problem. We now specify the extended HJB system corresponding the value function from Definition \ref{def_value_fct}.

\begin{definition} For $(t,p,z,h) \in [0,T] \times \mathbb{R}^{+} \times \mathbb{R}^{m+2}_{+} \times \mathbb{R}^{l}, $

\begin{align}
\begin{split}
\dot{V}(t,p,z,h) + \sup_{u \in U}\ \{\mathcal{L}^{u}V(t,p,z,h) + \mathcal{G}^{u}g_{u}(t,p,z,h)\} &= 0, \\
V(T,p,z,h) &= p-D(h), \\
\mathcal{A}^{\hat{u}}g_{\hat{u}}(t,p,z,h) & =0, \\
g_{\hat{u}}(T,p,z,h) = \mathbb{E}_{T,p,z,h}[P_{T}^{\hat{u}}-D(H_{T})] &=p-D(h),
\end{split}
\label{eq:ext_hjb}
\end{align}
where $\hat{u} = \argsup_{u \in U}\{\mathcal{L}^{u}V(t,p,z,h) + \mathcal{G}^{u}g_{u}(t,p,z,h)\}.$
\label{ext_hjb}
\end{definition}
\noindent
A solution to the extended HJB-system is the quadruple\\ $(\hat{u}, V(t,p,h,z), F_{\hat{u}}(t,p,h,z), g_{\hat{u}}(t,p,h,z)).$ 

\section{Sufficiency and Necessity}
\noindent
Before proving two verification results, we need the following assumption:

\begin{assumption}
The limit $\mathcal{A}^{u^{\star}}V(t,p,z,h)$ defined in \eqref{eq:def_gen} exists.
\end{assumption}
%\begin{assumption}
%For an equilibrium policy $u^{\star}$ of feedback type, the corresponding functions $F_{u^{\star}}$ and $g_{u^{\star}}$ given by \eqref{eq:aux_fct} satisfy $F_{u^{\star}}(t,p,z,h), g_{u^{\star}}(t,p,z,h) \in C^{1,2,2,2}([0,T) \times \mathbb{R}^{+} \times \mathbb{R}^{m}_{+} \times \mathbb{R}^{l}, \mathbb{R}) \cap C^{0,0,0,0}([0,T] \times \mathbb{R}^{+} \times \mathbb{R}^{m}_{+} \times \mathbb{R}^{l}, \mathbb{R}).$
%\label{ass_reg}
%\end{assumption}
\begin{definition}
A \emph{regular equilibrium} is a quadruple $(u^{\star}, V(t,p,z,h), F_{u^{\star}}(t,p,z,h),\\ g_{u^{\star}}(t,p,z,h)),$ with $u^{\star}$ being an equilibrium control of feedback type with corresponding equilibrium value function $V$ (cf. Definition \ref{def_equ}). % and $F_{u^{\star}$ and $g_{u^{\star}}$ satisfy Assumption \ref{ass_reg}.
\label{def_reg}
\end{definition}
The first verification theorem says that if the extended HJB-system given by Definition \ref{ext_hjb} has a solution, then it must be the equilibrium control law for the mean-variance hedge. In other words, the solvability of the extended HJB-system is \textit{sufficient} for the existence of an equilibrium control.

\begin{theorem}
Suppose $F_{u^{\star}}(t,p,z,h), g_{u^{\star}}(t,p,z,h) \in C^{1,2,2,2}([0,T) \times \mathbb{R}^{+} \times \mathbb{R}^{m}_{+} \times \mathbb{R}^{l}, \mathbb{R}) \cap C^{0,0,0,0}([0,T] \times \mathbb{R}^{+} \times \mathbb{R}^{m}_{+} \times \mathbb{R}^{l}, \mathbb{R})$ and $V, F_{u^{\star}}, g_{u^{\star}}$ solve the extended HJB-system in Definition \ref{ext_hjb}. Assume the control law $u^{\star}$ realizes the supremum in the first row for every quadruple $(t,p,z,h) \in [0,T] \times \mathbb{R}_{+} \times \mathbb{R}^{m+2}_{+} \times \mathbb{R}^{l}$. Then there exists an equilibrium control law $u^{\star}$ in the sense of Definition \ref{def_equ} and it is given by the optimal $u$ in the first row of \eqref{eq:ext_hjb}. Moreover, $V$ is the corresponding equilibrium value function and $F_{u^{\star}}$ and $g_{u^{\star}}$ are given by \eqref{eq:aux_fct}.
\end{theorem}
\begin{proof}
The proof can be conducted similarly to the proof of Theorem 7.1 in \cite{bm} and is therefore omitted.
\end{proof}
\noindent
Next we show that an equilibrium control is \textit{necessarily} a solution of the extended HJB-system. Such a proof is provided in \cite{lindensjo} for a diffusion case and we extend it to the present jump-diffusion setting including the hedge of the terminal condition. 
\begin{theorem}
A regular equilibrium $(u^{\star}, V(t,p,z,h), F_{u^{\star}}(t,p,z,h), g_{u^{\star}}(t,p,z,h))$ in the sense of Definiton \ref{def_reg} necessarily solves the extended HJB-system \eqref{eq:ext_hjb} and $u^{\star}$ realizes the supremum in the first row.
\label{thm_nec}
\end{theorem}
\noindent

The proof is delivered in several steps. 
We start by introducing two sequences of stopping times that will be repeatedly needed in the sequel. Let $(c_{k})_{k \in \mathbb{N}}$ be a strictly positive monotone sequence satisfying $\lim_{k \to \infty} c_{k} = 0.$ Let $(t,p,z,h,u)\in [0,T) \times \mathbb{R}_{+} \times \mathbb{R}^{m+2}_{+} \times \mathbb{R}^{l} \times \mathcal{U}$ arbitrary and denote by $B_{p},B_{z},B_{h}$ balls centered at respectively $p,z,h.$ Define the sequence of stopping times $(\sigma_{k}^{u})$ by

\begin{equation}
\sigma_{k}^{u} := \inf\{s > t: (s,P_{s}^{u},Z_{s},H_{s}) \notin [t,t+c_{k}) \times B_{p} \times B_{z} \times B_{h}\} \wedge T.
\label{eq:typ_el}
\end{equation}
\begin{proposition}
Consider the sequence of stopping times $(\sigma_{k}^{u})$ with a typical element given by \eqref{eq:typ_el}. It holds that
$$\sigma_{k}^{u} > t\ \text{a.s.}$$
\end{proposition}
\begin{proof}
This is an immediate consequence of the c\`{a}dl\`{a}g property of the mapping\\ $t \mapsto (t,P_{t}^{u},Z_{t},H_{t})$ (cf. \cite{applebaum}, p.106): let $k \in \mathbb{N}$ and $\omega \in \Omega$ arbitrary. For any $\epsilon > 0$ there exists some $ \delta > 0$ such that for all $s \in (t,t+\delta)$ it holds that
$$||(s,P^{u}_{s},Z_{s},H_{s}) - (t,P_{t}^{u},Z_{t},H_{t})||_{2} < \epsilon,$$
thus, $\sigma_{k}^{u} \geq s > t$ a.s.
%For any $\omega \in \Omega$ there exists some $L \in [t,t+c_{k}) \times B_{p} \times B_{z} \times B_{h}$ such that for any $\epsilon > 0\ \exists\ \delta > 0$ such that if  $t < s < t+\delta$ it holds that
%$$  ||(s,P_{s}^{u},Z_{s},H_{s}) - L|| < \epsilon,$$
%thus, $\sigma_{k}^{u} \geq t + \delta > t$ a.s.
\end{proof}
Observe that $\lim_{k \to \infty} \sigma_{k}^{u} = t.$ Let $(a_{k})_{k \in \mathbb{N}}$ be another positive monotone sequence satisfying $\lim_{k \to \infty} a_{k} = 0$ such that the sequence of events $(A_{k})_{k \in \mathbb{N}}$ is characterized by 
\begin{align}
\begin{split}
A_{k} &:= \{\omega \in \Omega: \sigma_{k}^{u} > t + a_{k}\}, \\
\mathbb{P}(A_{k}) &\geq 1- \frac{1}{k^{2}}.
\label{eq:A_{k}}
\end{split}
\end{align} 

\begin{lemma}
Consider the event $A_{k}$ and its probability of occurrence defined by \eqref{eq:A_{k}}. Then it holds that
$$\mathbbm{1}_{A_{k}}(\omega) = 1\ \text{a.s.},$$ for all but finitely many $k$.
\label{lem_set}
\end{lemma}
\begin{proof}
Observe that $\mathbb{P}(A_{k}^{c}) \leq \frac{1}{k^{2}}$ and therefore $\sum_{k=1}^{\infty} \mathbb{P}(A_{k}^{c}) \leq \frac{\pi^{2}}{6} < \infty.$ The \emph{Borel-Cantelli lemma} implies that $\mathbbm{1}_{A_{k}^{c}}(\omega) = 0$ for all but finitely many $k$ and the claim follows.
\end{proof}
\noindent
Define a typical element of the sequence of stopping times $(\tau_{k}^{u})_{k \in \mathbb{N}}$ by
\begin{equation}
\tau_{k}^{u} := \min \{\sigma_{k}^{u}, t+ a_{k}\}.
\label{eq:tau_k}
\end{equation}
\begin{lemma}
Let  $u^{\star}$ be an equilibrium control and consider the function $g_{u}$ defined by \eqref{eq:aux_fct}. Then it holds that
\begin{equation}
\mathcal{A}^{u^{\star}}g_{u^{\star}}(t,p,z,h) = 0.
\label{eq:kolmog}
\end{equation}
\label{lem_kol}
\end{lemma}
\begin{proof}
Using \emph{Dynkin's Formula} (see e.g. \cite{oksjump}, p.12), we find that
\begin{align*}
g&_{u^{\star}}(t,p,z,h) \\
&= \mathbb{E}_{t,p,z,h}\left[g_{u^{\star}}(\tau_{k}^{u^{\star}},P^{u^{\star}}_{\tau_{k}^{u^{\star}}},Z_{\tau_{k}^{u^{\star}}},H_{\tau_{k}^{u^{\star}}}) - \int_{t}^{\tau_{k}^{u^{\star}}}\mathcal{A}^{u^{\star}} g_{u^{\star}}(s,P_{s}^{u^{\star}},Z_{s},H_{s}) \ ds\right].
%\label{eq:aux1}
\end{align*}

It is a simple consequence of the \emph{Tower Property} that
\begin{align*}
\mathbb{E}&_{t,p,z,h}[g_{u^{\star}}(\tau_{k}^{u^{\star}},P^{u^{\star}}_{\tau_{k}^{u^{\star}}},Z_{\tau_{k}^{u^{\star}}},H_{\tau_{k}^{u^{\star}}})] \\
&= \mathbb{E}_{t,p,z,h}\left[\mathbb{E}_{\tau_{k}^{u^{\star}},P^{u^{\star}}_{\tau_{k}^{u^{\star}}},Z_{\tau_{k}^{u^{\star}}},H_{\tau_{k}^{u^{\star}}}}[P_{T}^{u^{\star}}-D(H_{T})]\right] \\
&= \mathbb{E}_{t,p,z,h}[P_{T}^{u^{\star}}-D(H_{T})] = g_{u^{\star}}(t,p,z,h).
\end{align*}

Combining the previous two results, we find that

$$\mathbb{E}_{t,p,z,h}\left[\frac{\int_{t}^{\tau_{k}^{u^{\star}}} \mathcal{A}^{u^{\star}} g_{u^{\star}}(s,P_{s}^{u^{\star}},Z_{s},H_{s}) \ ds}{a_{k}} \right] = 0.$$

Consider the sequence of random variables $$\left(\frac{\int_{t}^{\tau_{k}^{u^{\star}}} \mathcal{A}^{u^{\star}} g_{u^{\star}}(s,P_{s}^{u^{\star}},Z_{s},H_{s}) \ ds}{a_{k}} \right)_{k \in \mathbb{N}},$$ and note that the integrand is bounded on the interval $[t,\tau_{k}^{u^{\star}}],$ even if there is a large jump at $\tau_{k}^{u^{\star}}$ since this point has Lebesgue measure zero. Therefore we can use \emph{Dominated Convergence} to see that
\begingroup
\allowdisplaybreaks
	\begin{align*}
	\lim_{k \to \infty}\ &\mathbb{E}_{t,p,z,h}\left[\frac{\int_{t}^{\tau_{k}^{u^{\star}}} \mathcal{A}^{u^{\star}} g_{u^{\star}}(s,P_{s}^{u^{\star}},Z_{s},H_{s}) \ ds}{a_{k}} \right] \\
	&= \mathbb{E}_{t,p,z,h} \left[\lim_{k \to \infty} \frac{\int_{t}^{\tau_{k}^{u^{\star}}} \mathcal{A}^{u^{\star}} g_{u^{\star}}(s,P_{s}^{u^{\star}},Z_{s},H_{s}) \ ds}{a_{k}}  \right] \\
	&= \mathbb{E}_{t,p,z,h} \left[\lim_{k \to \infty} \mathbbm{1}_{A_{k}}(\omega)\ \frac{\int_{t}^{t+a_{k}} \mathcal{A}^{u^{\star}} g_{u^{\star}}(s,P_{s}^{u^{\star}},Z_{s},H_{s}) \ ds}{a_{k}} \right] \\
		&\ \ + \mathbb{E}_{t,p,z,h} \left[\lim_{k \to \infty} \mathbbm{1}_{A_{k}^{c}}(\omega)\ \frac{\int_{t}^{\sigma_{k}^{u^{\star}}} \mathcal{A}^{u^{\star}} g_{u^{\star}}(s,P_{s}^{u^{\star}},Z_{s},H_{s}) \ ds}{a_{k}}\right].
	%= &\mathbb{E}_{t,p,z,h} \left[\mathcal{A}^{u^{\star}}_{t,p,z,h}g^{n}_{u^{\star}}(t,p,z,h) \right] 
	\end{align*}
	\endgroup
According to Lemma \ref{lem_set} we have for arbitrary but fixed $\omega \in \Omega$ that $\mathbbm{1}_{A_{k}^{c}}(\omega) \neq 0$ for only finitely many $k$, therefore 
$$\lim_{k \to \infty} \mathbbm{1}_{A_{k}^{c}}(\omega) \frac{\int_{t}^{\sigma_{k}^{u^{\star}}} \mathcal{A}^{u^{\star}} g_{u^{\star}}(s,P_{s}^{u^{\star}},Z_{s},H_{s}) \ ds}{a_{k}} =0.$$
Further,  
\begin{align*}
0 &= \mathbb{E}_{t,p,z,h} \left[\lim_{k \to \infty} \frac{\int_{t}^{t+a_{k}} \mathcal{A}^{u^{\star}} g_{u^{\star}}(s,P_{s}^{u^{\star}},Z_{s},H_{s}) \ ds}{a_{k}} \right] \\
&= \mathbb{E}_{t,p,z,h} \left[\mathcal{A}^{u^{\star}}g_{u^{\star}}(t,p,z,h) \right] = \mathcal{A}^{u^{\star}}g_{u^{\star}}(t,p,z,h),
\end{align*}
whereby the second equality is justified by \emph{Lebesgue's Differentiation Theorem} (cf. \cite{rudin}, Chapter 7) and since $(t,p,z,h)$ has been arbitrarily chosen, \eqref{eq:kolmog} is established.
\end{proof}

Let $\tilde{u}_{\tau_{k}^{u}}$ be an allocation rule that is equal to $u(t) \equiv u \in U$ (a constant) on the interval $[t,\tau_{k}^{u}]$ and equal to the equilibrium $u^{\star}$ outside that interval, that is
\begin{align}
\tilde{u}_{\tau_{k}^{u}}(s) &= u\ \mathbbm{1}_{[t,\tau_{k}^{u})}(s) + u^{\star}(s) \ \mathbbm{1}_{[\tau_{k}^{u},T]}(s).
\label{eq:aux_stopping} \\
&= \left(u\ \mathbbm{1}_{[t,\sigma_{k}^{u})}(s) + u^{\star}(s) \ \mathbbm{1}_{[\sigma_{k}^{u},T]}(s) \right) \mathbbm{1}_{A_{k}^{c}}(\omega)  + \left(u\ \mathbbm{1}_{[t,t+a_{k})}(s) + u^{\star}(s) \ \mathbbm{1}_{[t+a_{k},T]}(s) \right) \mathbbm{1}_{A_{k}}(\omega) \label{eq:aux_stopping2}\\
&= \left(u\ \mathbbm{1}_{[t,\sigma_{k}^{u})}(s) + u^{\star}(s) \ \mathbbm{1}_{[\sigma_{k}^{u},T]}(s) \right) \mathbbm{1}_{A_{k}^{c}}(\omega)  + u_{t+a_{k}}(s) \mathbbm{1}_{A_{k}}(\omega) \label{eq:aux_stopping3},
\end{align}
where \eqref{eq:aux_stopping2} follows from the definition of $\tau_{k}^{u}$, cf. \eqref{eq:tau_k}. Moreover, as the right-hand bracket of \eqref{eq:aux_stopping2} is for each fixed $k \in \mathbb{N}$ easily seen to be a function of feedback type as $u_{t+c}$ in Definition \ref{def_equ}, we can equate it to \eqref{eq:aux_stopping3}.

\begin{lemma}
Consider an equilibrium control $u^{\star}$, the control $\tilde{u}_{\tau_{k}^{u}}$ given by \eqref{eq:aux_stopping} and the function $F_{u}$ defined by \eqref{eq:aux_fct}. Then we have
\begin{equation}
\lim_{k \to \infty} \frac{F_{u^{\star}}(t,p,z,h)-F_{\tilde{u}_{\tau_{k}^{u}}}(t,p,z,h)}{a_{k}} = - \mathcal{A}^{u^{\star}}F_{u^{\star}}(t,p,z,h).
\end{equation}
\end{lemma}
\begin{proof}
According to \emph{Dynkin's Formula}, 
\begin{align*}
\mathbb{E}&_{t,p,z,h}[F_{u^{\star}}(\tau_{k}^{u},P_{\tau_{k}^{u}}^{\tilde{u}_{\tau_{k}^{u}}},Z_{\tau_{k}^{u}},H_{\tau_{k}^{u}})] \\
&= F_{u^{\star}}(t,p,z,h) + \mathbb{E}_{t,p,z,h} \left[\int_{t}^{\tau_{k}^{u}} \mathcal{A}^{\tilde{u}_{\tau_{k}^{u}}} F_{u^{\star}}(s,P_{s}^{\tilde{u}_{\tau_{k}^{u}}},Z_{s},H_{s})\ ds \right],
\end{align*}
and we observe that
\begin{itemize}
	\item the integral limits in the previous equation are $t$ and $\tau_{k}^{u},$ therefore we can denote $P_{s}^{\tilde{u}_{\tau_{k}^{u}}}$ by $P_{s}^{u}$ and $\mathcal{A}^{\tilde{u}_{\tau_{k}^{u}}}$ by $\mathcal{A}^{u}$ on the random interval $(t,\tau_{k}^{u}).$
	\item as the starting time point is $\tau_{k}^{u},$ it holds that
	\begin{align*}
	F&_{u^{\star}}(\tau_{k}^{u},P_{\tau_{k}^{u}}^{\tilde{u}_{\tau_{k}^{u}}},Z_{\tau_{k}^{u}},H_{\tau_{k}^{u}}) \\
	&= \mathbb{E}_{\tau_{k}^{u},P_{\tau_{k}^{u}}^{\tilde{u}_{\tau_{k}^{u}}},Z_{\tau_{k}^{u}},H_{\tau_{k}^{u}}}\left[P_{T}^{u^{\star}} - \frac{\gamma}{2} (P_{T}^{u^{\star}})^{2} + \gamma P_{T}^{u^{\star}} D(H_{T}) - D(H_{T}) - \frac{\gamma}{2} D(H_{T})^{2}\right] \\
	&= \mathbb{E}_{\tau_{k}^{u},P_{\tau_{k}^{u}}^{\tilde{u}_{\tau_{k}^{u}}},Z_{\tau_{k}^{u}},H_{\tau_{k}^{u}}}\left[P_{T}^{\tilde{u}_{\tau_{k}^{u}}} - \frac{\gamma}{2} (P_{T}^{\tilde{u}_{\tau_{k}^{u}}})^{2} + \gamma P_{T}^{\tilde{u}_{\tau_{k}^{u}}} D(H_{T}) - D(H_{T}) - \frac{\gamma}{2} D(H_{T})^{2}\right].
	\end{align*}
\end{itemize}
Using these two observations, we rewrite
\begin{align*}
F&_{u^{\star}}(t,p,z,h) + \mathbb{E}_{t,p,z,h} \left[\int_{t}^{\tau_{k}^{u}} \mathcal{A}^{u} F_{u^{\star}}(s,P_{s}^{u},Z_{s},H_{s})\ ds \right] \\
&= \mathbb{E}_{t,p,z,h}[F_{u^{\star}}(\tau_{k}^{u},P_{\tau_{k}^{u}}^{\tilde{u}_{\tau_{k}^{u}}},Z_{\tau_{k}^{u}},H_{\tau_{k}^{u}})]\\
&= \mathbb{E}_{t,p,z,h}\left[\mathbb{E}_{\tau_{k}^{u},P_{\tau_{k}^{u}}^{\tilde{u}_{\tau_{k}^{u}}},Z_{\tau_{k}^{u}},H_{\tau_{k}^{u}}}\left[P_{T}^{\tilde{u}_{\tau_{k}^{u}}} - \frac{\gamma}{2} (P_{T}^{\tilde{u}_{\tau_{k}^{u}}})^{2} + \gamma P_{T}^{\tilde{u}_{\tau_{k}^{u}}} D(H_{T}) - D(H_{T}) - \frac{\gamma}{2} D(H_{T})^{2}\right] \right] \\
&=\mathbb{E}_{t,p,z,h}\left[P_{T}^{\tilde{u}_{\tau_{k}^{u}}} - \frac{\gamma}{2} (P_{T}^{\tilde{u}_{\tau_{k}^{u}}})^{2} + \gamma P_{T}^{\tilde{u}_{\tau_{k}^{u}}} D(H_{T}) - D(H_{T}) - \frac{\gamma}{2} D(H_{T})^{2}\right] \\
&=F_{\tilde{u}_{\tau_{k}^{u}}}(t,p,z,h).
\end{align*}
Finally, we use \emph{Dominated Convergence} and \emph{Lebesgue's Differentiation Theorem} similarly as in the proof of Lemma \ref{lem_kol} to deduce that
\begin{align*}
&\lim_{k \to \infty} \frac{F_{u^{\star}}(t,p,z,h)-F_{\tilde{u}_{\tau_{k}^{u}}}(t,p,z,h)}{a_{k}} \\
&= \lim_{k \to \infty} \frac{-\mathbb{E}_{t,p,z,h} \left[\mathbbm{1}_{A_{k}}(\omega)\int_{t}^{t+a_{k}} \mathcal{A}^{u} F_{u^{\star}}(s,P_{s}^{u},Z_{s},H_{s})\ ds \right]}{a_{k}} \\
&= \mathbb{E}_{t,p,z,h}\left[\lim_{k \to \infty} \mathbbm{1}_{A_{k}}(\omega) \frac{-\int_{t}^{t+a_{k}} \mathcal{A}^{u} F_{u^{\star}}(s,P_{s}^{u},Z_{s},H_{s})\ ds}{a_{k}}\right]\\
&= - \mathcal{A}_{t,p,z,h}^{u^{\star}}F_{u^{\star}}(t,p,z,h),
\end{align*}
which is what we have set out to prove.
\end{proof}
%\newpage
\begin{lemma}
Consider an equilibrium control $u^{\star}$, the control $\tilde{u}_{\tau_{k}^{u}}$ given by \eqref{eq:aux_stopping} and the function $g_{u}$ defined by \eqref{eq:aux_fct}. Then we have
\begin{equation}
\lim_{k \to \infty} \frac{g_{u^{\star}}(t,p,z,h)^{2}-g_{\tilde{u}_{\tau_{k}^{u}}}(t,p,z,h)^{2}}{a_{k}} = - 2\ g_{u^{\star}}(t,p,z,h) \ \mathcal{A}^{u^{\star}} g_{u^{\star}}(t,p,z,h).
\end{equation}
\end{lemma}
\begin{proof}
Using similar techniques as before, the following calculation yields
\begingroup
\allowdisplaybreaks
\begin{align*}
&\lim_{k \to \infty} \frac{g_{u^{\star}}(t,p,z,h)^{2}-g_{\tilde{u}_{\tau_{k}^{u}}}(t,p,z,h)^{2}}{a_{k}} \\
&= - \lim_{k \to \infty} \frac{g_{\tilde{u}_{\tau_{k}^{u}}}(t,p,z,h)^{2}-g_{u^{\star}}(t,p,z,h)^{2}}{a_{k}} \\
&=- \lim_{k \to \infty} \frac{\left(\mathbb{E}_{t,p,z,h}[P_{T}^{\tilde{u}_{\tau_{k}^{u}}}-D(H_{T})]\right)^{2}-g_{u^{\star}}(t,p,z,h)^{2}}{a_{k}} \\
&\stackrel{\text{T.P.}}{=} - \lim_{k \to \infty} \frac{\left(\mathbb{E}_{t,p,z,h}\left[\mathbb{E}_{\tau_{k},P_{\tau_{k}}^{\tilde{u}_{\tau_{k}^{u}}},Z_{\tau_{k}},H_{\tau_{k}}}\left[P_{T}^{\tilde{u}_{\tau_{k}^{u}}}-D(H_{T})\right]\right]\right)^{2}-g_{u^{\star}}(t,p,z,h)^{2}}{a_{k}} \\
&= - \lim_{k \to \infty} \frac{\left(\mathbb{E}_{t,p,z,h}\left[\mathbb{E}_{\tau_{k},P_{\tau_{k}}^{\tilde{u}_{\tau_{k}^{u}}},Z_{\tau_{k}},H_{\tau_{k}}}\left[P_{T}^{u^{\star}}-D(H_{T})\right]\right]\right)^{2}-g_{u^{\star}}(t,p,z,h)^{2}}{a_{k}}\\
& = - \lim_{k \to \infty} \frac{\left(\mathbb{E}_{t,p,z,h}\left[g_{u^{\star}}(\tau_{k},P_{\tau_{k}}^{\tilde{u}_{\tau_{k}^{u}}},Z_{\tau_{k}},H_{\tau_{k}})\right]\right)^{2}-g_{u^{\star}}(t,p,z,h)^{2}}{a_{k}} \\
&= - \lim_{k \to \infty} \frac{\left(g_{u^{\star}}(t,p,z,h) + \mathbb{E}_{t,p,z,h} \left[\int_{t}^{\tau_{k}^{u}} \mathcal{A}^{u} g_{u^{\star}}(s,P_{s}^{u},Z_{s},H_{s})\ ds \right]\right)^{2}-g_{u^{\star}}(t,p,z,h)^{2}}{a_{k}} \\
&= - 2\ g_{u^{\star}}(t,p,z,h) \ \mathcal{A}^{u^{\star}} g_{u^{\star}}(t,p,z,h),
\end{align*}
\endgroup
where the abbreviation T.P. stands for \emph{Tower Property}.
\end{proof}

\begin{lemma}
Consider an equilibrium control $u^{\star}$, the control $\tilde{u}_{\tau_{k}^{u}}$ given by \eqref{eq:aux_stopping} and the value function $J$ specified in Definition \ref{def_value_fct}. Then it holds that
$$- \lim_{k \to \infty} \frac{J(t,p,z,h,u^{\star}) - J(t,p,z,h,\tilde{u}_{\tau_{k}^{u}})}{a_{k}} = \mathcal{A}^{u^{\star}}V(t,p,z,h)+ \mathcal{G}^{u^{\star}}g_{u^{\star}}(t,p,z,h) .$$
\label{lem_inequ}
\end{lemma}
\begin{proof}
\begin{align*}
 &- \lim_{k \to \infty} \frac{J(t,p,z,h,u^{\star}) - J(t,p,z,h,\tilde{u}_{\tau_{k}^{u}})}{a_{k}} \\
&= - \lim_{k \to \infty} \frac{F_{u^{\star}}(t,p,z,h)-F_{\tilde{u}_{\tau_{k}^{u}}}(t,p,z,h) + \frac{\gamma}{2} \left(g_{u^{\star}}(t,p,z,h)^{2}-g_{\tilde{u}_{\tau_{k}^{u}}}(t,p,z,h)^{2} \right)}{a_{k}} \\
&= \mathcal{A}^{u^{\star}} F_{u^{\star}}(t,p,z,h) + \gamma g_{u^{\star}}(t,p,z,h) \mathcal{A}^{u^{\star}} g_{u^{\star}}(t,p,z,h) \\
&= \dot{F}_{u^{\star}}(t,p,z,h) + \mathcal{L}^{u^{\star}} F_{u^{\star}}(t,p,z,h) + \gamma g_{u^{\star}}(t,p,z,h) \left(\dot{g}_{u^{\star}}(t,p,z,h) + \mathcal{L}^{u^{\star}} g_{u^{\star}}(t,p,z,h)\right) \\
&= \underbrace{\dot{F}_{u^{\star}}(t,p,z,h) + \gamma g_{u^{\star}}(t,p,z,h) \dot{g}_{u^{\star}}(t,p,z,h)}_{=\dot{V}(t,p,z,h)} + \underbrace{\mathcal{L}^{u^{\star}} F_{u^{\star}}(t,p,z,h) + \frac{\gamma}{2} \mathcal{L}^{u^{\star}}g_{u^{\star}}(t,p,z,h)^{2}}_{=\mathcal{L}^{u^{\star}}V(t,p,z,h)} \\
& + \underbrace{\gamma g_{u^{\star}}(t,p,z,h) \mathcal{L}^{u^{\star}} g_{u^{\star}}(t,p,z,h) -  \frac{\gamma}{2} \mathcal{L}^{u^{\star}}g_{u^{\star}}(t,p,z,h)^{2}}_{= \mathcal{G}^{u^{\star}}g_{u^{\star}}(t,p,z,h)}.
\end{align*}
\end{proof}

\begin{proof}[Proof of Theorem \ref{thm_nec}] The proof is conducted in four steps:\\

\noindent
\textsc{Step $1$:} We show the boundary conditions.\\
The boundary conditions $V(T,p,z,h) = p-D(h)$ and $g_{u^{\star}}(T,p,z,h) = p-D(h)$ are met by the equilibrium control law $u^{\star}$, which follows from Definition \ref{aux_fct} and Definition \ref{def_equ}.\\

\noindent
\textsc{Step $2$:} Observe that $\mathcal{A}^{u^{\star}}g_{u^{\star}}(t,p,z,h) = 0$ is stated by Lemma \ref{lem_kol}.\\

\noindent
\textsc{Step $3$:} We show that $\mathcal{A}^{u^{\star}}V(t,p,z,h) + \mathcal{G}^{u^{\star}}g_{u^{\star}}(t,p,z,h) = 0.$ \\
Recall from Definition \ref{def_value_fct} and Definition \ref{def_equ} that $V(t,p,z,h) = F_{u^{\star}}(t,p,z,h) + \frac{\gamma}{2} g_{u^{\star}}^{2}(t,p,z,h).$ Following a similar line of reasoning as in the proof of Lemma \ref{lem_kol}, one can show that $\mathcal{A}^{u^{\star}}F_{u^{\star}}(t,p,z,h) = 0.$ So we have
\begin{align*}
\mathcal{A}&^{u^{\star}}V(t,p,z,h) + \mathcal{G}^{u^{\star}}g_{u^{\star}}(t,p,z,h) \\
&= \frac{\gamma}{2} \mathcal{A}^{u^{\star}}g_{u^{\star}}^{2}(t,p,z,h) + \gamma g_{u^{\star}}(t,p,z,h) \mathcal{L}^{u^{\star}} g_{u^{\star}}(t,p,z,h) - \frac{\gamma}{2} \mathcal{L}^{u^{\star}} g_{u^{\star}}^{2}(t,p,z,h) \\
&= \frac{\gamma}{2} \mathcal{A}^{u^{\star}}g_{u^{\star}}^{2}(t,p,z,h) + \gamma g_{u^{\star}}(t,p,z,h) (\mathcal{A}^{u^{\star}} g_{u^{\star}}(t,p,z,h) - \dot{g}_{u^{\star}}(t,p,z,h)) \\
&\ \ - \frac{\gamma}{2} (\mathcal{A}^{u^{\star}} g_{u^{\star}}^{2}(t,p,z,h) - 2g_{u^{\star}}(t,p,z,h) \dot{g}_{u^{\star}}(t,p,z,h)) \\
&= \gamma g_{u^{\star}}(t,p,z,h)\ \mathcal{A}^{u^{\star}} g_{u^{\star}}(t,p,z,h) =0,
\end{align*}
whereby the last equality follows from \textsc{Step 2}.\\

 So far, we have shown that the regular equilibrium $(u^{\star},V(t,p,z,h),F_{u^{\star}}(t,p,z,h),\\ g_{u^{\star}}(t,p,z,h))$ is a prospective solution of the extended HJB-system \eqref{eq:ext_hjb}. Therefore we are left showing that $u^{\star}$ is indeed maximal in the first row of \eqref{eq:ext_hjb}.\\

\noindent
\textsc{Step 4:} We show that $0 \geq \mathcal{A}^{u^{\star}}V(t,p,z,h) + \mathcal{G}^{u^{\star}}g_{u^{\star}}(t,p,z,h).$\\
In the following calculation, the first inequality follows by Definition of the equilibrium control $u^{\star}$, cf. Definition \ref{def_equ}.
Observe that
\begingroup
\allowdisplaybreaks
\begin{align*}
0 & \geq - \liminf_{c \searrow 0} \frac{J(t,p,z,h,u^{\star})-J(t,p,z,h,u_{t+c})}{c} \\
&= - \liminf_{c \searrow 0} \Bigg(\frac{J(t,p,z,h,u^{\star})-J(t,p,z,h,u_{t+c})}{c} \mathbbm{1}_{A_{k}}(\omega) \\
&\ + \frac{J(t,p,z,h,u^{\star})-J(t,p,z,h,u_{t+c})}{c} \mathbbm{1}_{A_{k}^{c}}(\omega) \Bigg) \\
%&= - \lim_{c \searrow 0} \lim_{n \to \infty} \frac{J^{n}(t,p,z,h,u^{\star})-J^{n}(t,p,z,h,u_{t+c})}{c} \\
&= - \liminf_{k \to \infty} \frac{J(t,p,z,h,u^{\star})-J(t,p,z,h,u_{t+a_{k}})}{a_{k}} \mathbbm{1}_{A_{k}}(\omega)\\
&\ - \liminf_{k \to \infty} \frac{J(t,p,z,h,u^{\star})-J(t,p,z,h,u_{t+c_{k}})}{c_{k}} \mathbbm{1}_{A_{k}^{c}}(\omega).
\end{align*}
\endgroup
Note that Lemma \ref{lem_set} implies that 
$$\liminf_{k \to \infty} \frac{J(t,p,z,h,u^{\star})-J(t,p,z,h,u_{t+c_{k}})}{c_{k}} \mathbbm{1}_{A_{k}^{c}}(\omega) = 0.$$
Consequently, we deduce from \eqref{eq:aux_stopping3} that

\begin{align*}
&- \liminf_{k \to \infty} \frac{J(t,p,z,h,u^{\star})-J(t,p,z,h,u_{t+a_{k}})}{a_{k}} \mathbbm{1}_{A_{k}}(\omega) \\
&= - \liminf_{k \to \infty} \frac{J(t,p,z,h,u^{\star}) - J(t,p,z,h,\tilde{u}_{\tau_{k}^{u}})}{a_{k}},
\end{align*} 
and Lemma \ref{lem_inequ} concludes the proof.
\end{proof}

\section{Explicit solution}

In the sequel, let $D \equiv 0$, i.e., we consider an investor aiming at receiving a high expected payoff while keeping its variance low. In this special case the extended HJB-system \eqref{eq:ext_hjb} admits explicit closed-form solutions. Some notational definitions are in order:
\begin{itemize}
	\item $\sigma_{S} := (\sigma_{ij})_{1 \leq i \leq m, \\ 1\leq j \leq d},$ i.e., $\sigma_{S} \in \mathbb{R}^{m \times d}.$
	\item $\tilde{\sigma}_{S} := \sigma_{S} \sigma_{S}^{\intercal},$ i.e., $\tilde{\sigma}_{S} \in \mathbb{R}^{m \times m}.$ Note that $\tilde{\sigma}_{S}$ is a symmetric matrix.  
	\item $\tilde{\sigma}_{S^{i}} := (\tilde{\sigma}_{S_{i1}},\dots,\tilde{\sigma}_{S_{im}})^{\intercal}$, i.e., $\tilde{\sigma}_{S^{i}} \in \mathbb{R}^{m}.$
	\item $\sigma_{S^{i}} := (\sigma_{i1},\dots,\sigma_{id})^{\intercal},$ i.e., $\sigma_{S}^{i} \in \mathbb{R}^{d}$ for every $i \in \{1,\dots,m\}.$
	\item $\rho_{S} := (\rho_{ij})_{1 \leq i \leq m, 1 \leq j \leq k},$ i.e., $\rho_{S} \in \mathbb{R}^{m \times k}.$
	\item $\tilde{\rho}_{S} := \rho_{S} \rho_{S}^{\intercal}.$
	\item $\rho_{S^{i}} := (\rho_{i1},\dots,\rho_{ik})^{\intercal},$ i.e., $\rho_{S^{i}} \in \mathbb{R}^{k}.$
	\item $\mu := (\mu_{1},\dots,\mu_{n})^{\intercal},$ i.e., $\mu \in \mathbb{R}^{n}.$
	\item $\tilde{\mu} := \mu - \textbf{1}r.$
	\item $\Delta P_{t}^{u}(x,\bar{x}) := u_{S}^{\intercal}(t) \rho_{S} x + u_{Y}(t) \eta_{L}(t,\lambda_{t-},Y_{t-},\bar{x}),$ i.e., $\Delta P^{u}_{t}(x,\bar{x}) \in \mathbb{R}.$
	\item $\Delta Z_{t}(x,\bar{x}) := (\text{Diag}(S_{t-})\rho_{S}x, \tilde{\sigma}_{\lambda}(t,\lambda_{t-},\bar{x}), Y_{t-} \eta_{L}(t,\lambda_{t-},Y_{t-})),$ i.e., $\Delta Z_{t}(x,\bar{x}) \in \mathbb{R}^{m+2}.$\\
	\item $\mu_{i}(t,Z_{t}) := \begin{cases} \mu_{i}S_{t}^{i},\ &i=1,\dots,m,\\ \mu_{\lambda}(t,\lambda_{t}), \ &i=m+1, \\ (r+\nu_{L}(t,\lambda_{t},Y_{t}))Y_{t}, \ &i=m+2, \end{cases}$\\
	i.e., $\mu(t,Z_{t}) = (\mu_{1}(t,Z_{t}),\dots,\mu_{m+2}(t,Z_{t})) \in \mathbb{R}^{m+2}.$\\
	\item $\sigma_{ij}(t,Z_{t}) := \begin{cases} \sigma_{ij}S_{t}^{i}, \ & 1\leq i\leq m, 1\leq j, \leq d, \\ \sigma_{\lambda}(t,\lambda_{t}), & i=m+1, j=d+1, \\ \sigma_{L}(t,\lambda_{t},Y_{t})Y_{t}, & i=m+2, j=d+1, \\ 0, & \text{else,} \end{cases}$ \\
	i.e., $\sigma(t,Z_{t}) \in \mathbb{R}^{(m+2) \times (d+1)}.$\\
	\item $Q^{u}_{i}(t,Z_{t}) := \begin{cases}S_{t}^{i} u_{S}^{\intercal}(t) \sigma_{S} \sigma_{S^{i}}, & 1 \leq i \leq m, \\ u_{Y}(t) \sigma_{\lambda}(t,\lambda_{t}) \sigma_{L}(t,\lambda_{t},Y_{t}), & i=m+1, \\ u_{Y}(t) \sigma_{L}^{2}(t,\lambda_{t},Y_{t}) Y_{t}, & i=m+2, \end{cases}$\\ i.e., $Q^{u}(t,Z_{t})=(Q^{u}_{1}(t,Z_{t}),\dots,Q^{u}_{m+2}(t,Z_{t})) \in \mathbb{R}^{m+2}.$
\end{itemize}

 Inspired by \cite{basak} and \cite{bm}, we make the following \emph{Ansatz}:
\begin{align}
\begin{split}
V(t,p,z) &= A(t) p + B(t,z), \\
g(t,p,z) &= a(t) p + b(t,z).
\end{split}
\label{eq:ansatz}
\end{align}
The goal is finding the functions $A,a,B,b$ as well as the equilibrium control laws of feedback type. Clearly, the functions $A,a,B,b$ are assumed to satisfy the necessary regularity conditions and the limits induced by applying the operators $\mathcal{A}, \mathcal{L}$ and $\mathcal{G}$ are assumed to exist accordingly. Consider the first line in the system \eqref{eq:ext_hjb} and define
$$\Xi(u_{S}(t), u_{Y}(t)) := \mathcal{L}^{u} V(t,p,z) + \mathcal{G}^{u}g(t,p,z).$$
\noindent
Omitting details at this stage, an application of It\^{o}'s formula for jump-diffusions yields 
\begin{align}
\begin{split}
&\Xi(u_{S}(t), u_{Y}(t)) \\
&= A(t) (pr+\tilde{\mu}^{\intercal} u_{S}(t) + u_{Y}(t) \nu_{L}(t,\lambda_{t},Y_{t})) + \nabla_{Z}B(t,Z_{t})^{\intercal} \mu(t,Z_{t})\\
&\ +  \frac{1}{2}\ \tr\Big(\sigma(t,Z_{t})^{\intercal} H_{Z}(B(t,Z_{t})) \sigma(t,Z_{t})  \Big)  - \frac{\gamma}{2}a(t)^{2} (u_{S}^{\intercal}(t) \tilde{\sigma}_{S} u_{S}(t) + u_{Y}^{2}(t) \sigma_{L}^{2}(t,\lambda_{t},Y_{t})) \\ & 
\ -\frac{\gamma}{2} \tr \left(\sigma(t,Z_{t})^{\intercal} \nabla_{Z}b(t,Z_{t})^{\intercal} \nabla_{Z}b(t,Z_{t}) \sigma(t,Z_{t}) \right) - \gamma a(t) \nabla_{Z}b(t,Z_{t})^{\intercal} Q^{u}(t,Z_{t}) \\
&\ + \int_{\mathbb{R}^{k+1} \setminus \{0\}} \left( B(t,Z_{t}+\Delta Z_{t}(x,\bar{x})) - B(t,Z_{t}) - (\Delta Z_{t}(x,\bar{x}))^{\intercal} \nabla_{Z}B(t,Z_{t})\right)\ \vartheta_{X,\bar{X}}(dx,d\bar{x}) \\
&\ - \frac{\gamma}{2} \int_{\mathbb{R}^{k+1} \setminus \{0\}} \left(a(t) \Delta P_{t}^{u}(x,\bar{x}) + b(t,Z_{t}+\Delta Z_{t}(x,\bar{x})) - b(t,Z_{t}) \right)^{2}\ \vartheta_{X,\bar{X}}(dx,d\bar{x}).
\end{split}
\label{eq:first_line}
\end{align}
Note that the maximization of \eqref{eq:first_line} w.r.t. $u_{S}(t)$ and $u_{Y}(t)$ is a static optimization problem in $m+1$ variables, so we solve the corresponding first order conditions (FOC). First, observe that
\begin{align*}
\frac{\partial u_{S}^{\intercal}(t) \tilde{\sigma}_{S} u_{S}(t)}{\partial u_{S^{i}}(t)} &= 2 \sum_{j=1}^{m} \tilde{\sigma}_{S_{ij}} u_{S^{j}}(t) = 2 \tilde{\sigma}_{S^{i}}^{\intercal} u_{S}(t) , \\
\frac{\partial u_{S}^{\intercal}(t) \rho_{S} x}{\partial u_{S^{i}}(t)} &= \sum_{j=1}^{k} \rho_{ij} x^{j} = \rho_{S^{i}}^{\intercal} x.
\end{align*}
For arbitrary $i \in \{1,\dots,m\}$, we consider the following FOC. Observe that the interchange of differentiation and integration is justified by our assumptions.
\begingroup
\allowdisplaybreaks
\begin{align*}
&\frac{\partial \Xi}{\partial u_{S^{i}}(t)}\\
 &= A(t)\tilde{\mu}_{i} - \gamma a(t)^{2} \tilde{\sigma}_{S^{i}}^{\intercal} u_{S}(t) -\gamma a(t) \sigma_{S^{i}}^{\intercal} \sum_{j=1}^{m} b_{z_{j}}(t,Z_{t}) S_{t}^{j} \sigma_{S^{j}} \\
&\ - \gamma a(t) \int_{\mathbb{R}^{k+1} \setminus \{0\}} \big(a(t) \Delta P_{t}^{u}(x,\bar{x}) + b(t,Z_{t}+\Delta Z_{t}(x,\bar{x})) - b(t,Z_{t})\big)\rho_{S^{i}}^{\intercal}x \ \vartheta_{X,\bar{X}}(dx,d\bar{x}) \\
&= A(t)\tilde{\mu}_{i} - \gamma a(t)^{2} \tilde{\sigma}_{S^{i}}^{\intercal} u_{S}(t) -\gamma a(t) \sigma_{S^{i}}^{\intercal} \sum_{j=1}^{m} b_{z_{j}}(t,z) S_{t}^{j} \sigma_{S^{j}} \\
&\ - \gamma a(t) \int_{\mathbb{R}^{k} \setminus \{0\}} \big(a(t) u_{S}^{\intercal}(t) \rho_{S} x + b(t,Z_{t}+\Delta Z_{t}(x,0)) - b(t,Z_{t})\big)\rho_{S^{i}}^{\intercal}x \ \vartheta_{X}(dx) \equiv 0 \\
& \Leftrightarrow A(t)\tilde{\mu}_{i} - \gamma a(t) \sigma_{S^{i}}^{\intercal} \sum_{j=1}^{m} b_{z_{j}}(t,z) S_{t}^{j} \sigma_{S^{j}} \\
&\ - \gamma a(t) \int_{\mathbb{R}^{k} \setminus \{0\}} (b(t,Z_{t}+\Delta Z_{t}(x,0)) - b(t,Z_{t}))\ x^{\intercal} \ \vartheta_{X}(dx)\ \rho_{S^{i}} \\
& \stackrel{*}{=} \gamma a(t)^{2}\left( \tilde{\sigma}_{S^{i}}^{\intercal} +\xi\ \rho_{S^{i}}^{\intercal}  \ \rho_{S}^{\intercal} \right) u_{S}(t).
\end{align*}
\endgroup
We use the following abbreviations in the sequel:
\begin{align}
\xi &:= \int_{\mathbb{R}^{k} \setminus \{0\}} x x^{\intercal}\ \vartheta_{X}(dx) \notag, \\
\tilde{\eta}_{L}(t,\lambda_{t},Y_{t}) &:= \int_{\mathbb{R} \setminus \{0\}} \eta_{L}(t,\lambda_{t},Y_{t},\bar{x})^{2} \ \vartheta_{\bar{X}}(d\bar{x}) \notag, \\
%\tilde{R} &:= RR^{\intercal}, \\
b_{1}(t,Z_{t}) &:= \int_{\mathbb{R}^{k} \setminus \{0\}} (b(t,Z_{t}+\Delta Z_{t}(x,0)) - b(t,Z_{t}))\ x \ \vartheta_{X}(dx), \label{eq:b1}\\
b_{2}(t,Z_{t}) &:= \int_{\mathbb{R} \setminus \{0\}} (b(t,Z_{t}+\Delta Z_{t}(0,\bar{x})) - b(t,Z_{t}))\ \eta_{L}(t,\lambda_{t},Y_{t},\bar{x}) \ \vartheta_{\bar{X}}(d\bar{x}). \label{eq:b2}
\end{align}
Note that in the optimum an equality of type $*$ needs to hold for every $u_{S^{i}}(t),$ so using the just defined functions and matrix-vector notation, we see that the vector $u_{S}^{\star}(t)$ has to satisfy
\begin{align}
u_{S}^{\star}(t) &= \frac{(\tilde{\sigma}_{S} +\tilde{\rho}_{S} \xi)^{-1}}{\gamma a(t)^{2}}\ \Big( A(t)\tilde{\mu}-\gamma a(t) \sigma_{S} \sum_{j=1}^{m} b_{z_{j}}(t,Z_{t}) S_{t}^{j} \sigma_{S^{j}} - \gamma a(t) \rho_{S} b_{1}(t,Z_{t}) \Big),
\label{eq:uS_prelim}
\end{align}
where the symbol $\star$ indicates the optimality of the strategy. Next we compute
\begin{align*}
&\frac{\partial \Xi}{\partial u_{Y}(t)} \\
&= A(t) \nu_{L}(t,\lambda_{t},Y_{t}) - \gamma a(t)^{2} \sigma_{L}^{2}(t,\lambda_{t},Y_{t}) u_{Y}(t) - \gamma a(t)\big(b_{z_{m+1}}(t,Z_{t}) \sigma_{\lambda}(t,\lambda_{t}) \sigma_{L}(t,\lambda_{t},Y_{t})\\
&\ + b_{z_{m+2}}(t,Z_{t}) \sigma^{2}_{L}(t,\lambda_{t},Y_{t}) Y_{t} \big) \\
&\ - \gamma a(t) \int_{\mathbb{R} \setminus \{0\}} \Big(a(t) u_{Y}(t) \eta_{L}(t,\lambda_{t},Y_{t},\bar{x}) + b(t,Z_{t}+\Delta Z_{t}(0,\bar{x}))-b(t,Z_{t})\Big)\eta_{L}(t,\lambda_{t},Y_{t},\bar{x})\ \vartheta_{\bar{X}}(d\bar{x})   \\
& = A(t) \nu_{L}(t,\lambda_{t},Y_{t}) - \gamma a(t)^{2} \sigma_{L}^{2}(t,\lambda_{t},Y_{t}) u_{Y}(t) - \gamma a(t)\big(b_{z_{m+1}}(t,Z_{t}) \sigma_{\lambda}(t,\lambda_{t}) \sigma_{L}(t,\lambda_{t},Y_{t})\\
&\ + b_{z_{m+2}}(t,Z_{t}) \sigma^{2}_{L}(t,\lambda_{t},Y_{t}) Y_{t} \big)- \gamma a(t)^{2} \tilde{\eta}_{L}(t,\lambda_{t},Y_{t}) u_{Y}(t) - \gamma a(t) b_{2}(t,Z_{t}) \equiv 0 \\
 &\Leftrightarrow u_{Y}^{\star}(t) = \frac{A(t) \nu_{L}(t,\lambda_{t},Y_{t}) -\gamma a(t) \big(b_{z_{m+1}}(t,Z_{t}) \sigma_{\lambda}(t,\lambda_{t}) \sigma_{L}(t,\lambda_{t},Y_{t}) + b_{z_{m+2}}(t,Z_{t}) \sigma^{2}_{L}(t,\lambda_{t},Y_{t}) Y_{t} \big)}{\gamma a(t)^{2} (\sigma_{L}^{2}(t,\lambda_{t},Y_{t}) + \tilde{\eta}_{L}(t,\lambda_{t},Y_{t}))} \\
&\ \ - \frac{\gamma a(t) b_{2}(t,Z_{t})}{\gamma a(t)^{2} (\sigma_{L}^{2}(t,\lambda_{t},Y_{t}) + \tilde{\eta}_{L}(t,\lambda_{t},Y_{t}))}.
\end{align*}

Observe that the optimal control does not depend on $p$.  We next plug $u^{\star}$ into \eqref{eq:first_line}. Then we can apply separation of variables to the first line of \eqref{eq:ext_hjb}. This leads to an ordinary differential equation (ODE) for $A$ and a PIDE for $B$. The ODE for $A$ is given by
\begin{align*}
\dot{A}(t) + A(t)r & =0, \\
A(T) & =1,
\end{align*}
and the solution is easily seen to be $A(t) = e^{r(T-t)}.$ The PIDE for $B$ is given by
\begin{align}
\begin{split}
&\dot{B}(t,Z_{t}) + A(t) \left(\tilde{\mu}^{\intercal} u_{S}^{\star}(t) + u_{Y}^{\star}(t) \nu_{L}(t,\lambda_{t},Y_{t}) \right) + \nabla_{Z}B(t,Z_{t})^{\intercal} \mu(t,Z_{t})\\
& +  \frac{1}{2}\ \tr\Big(\sigma(t,Z_{t})^{\intercal} H_{Z}(B(t,Z_{t})) \sigma(t,Z_{t})  \Big)  - \frac{\gamma}{2} \left(u_{S}^{\star}(t)^{\intercal} \tilde{\sigma}_{S} u_{S}^{\star}(t) + u_{Y}^{\star}(t)^{2} \sigma_{L}^{2}(t,\lambda_{t},Y_{t}) \right) a(t)^{2}\\
& -\frac{\gamma}{2} \tr \left(\sigma(t,Z_{t})^{\intercal} \nabla_{Z}b(t,Z_{t})^{\intercal} \nabla_{Z}b(t,Z_{t}) \sigma(t,Z_{t}) \right) - \gamma a(t) \nabla_{Z}b(t,Z_{t})^{\intercal} Q^{u^{\star}}(t,Z_{t}) \\
&+ \int_{\mathbb{R}^{k+1} \setminus \{0\}} \left( B(t,Z_{t}+\Delta Z_{t}(x,\bar{x})) - B(t,Z_{t}) - (\Delta Z_{t}(x,\bar{x}))^{\intercal} \nabla_{Z}B(t,Z_{t})\right)\ \vartheta_{X,\bar{X}}(dx,d\bar{x}) \\
&- \frac{\gamma}{2} \int_{\mathbb{R}^{k+1} \setminus \{0\}} \left(a(t) \Delta P_{t}^{u^{\star}}(x,\bar{x}) + b(t,Z_{t}+\Delta Z_{t}(x,\bar{x})) - b(t,Z_{t}) \right)^{2}\ \vartheta_{X,\bar{X}}(dx,d\bar{x}) =0.
\end{split}
\label{eq:static}
\end{align}
Note that $\Delta P_{t}^{u^{\star}}(x,\bar{x})$ in the last line of \eqref{eq:static} means the jump of the portfolio process where the investor is allocating optimally. For solving the latter PIDE, we need to find the functions $a$ and $b$. To do so, we use the third equation of the system \eqref{eq:ext_hjb} (for the special case $D \equiv 0$), namely $\mathcal{A}^{u^{\star}}g_{u^{\star}}(t,p,z) = 0.$ Following the \emph{Ansatz} $$g(t,p,z) = a(t)p + b(t,z),$$ we obtain

\begin{align}
\begin{split}
&\dot{a}(t)p + \dot{b}(t,Z_{t}) + a(t)(pr + \tilde{\mu}^{\intercal} u_{S}^{\star}(t) + \nu_{L}(t,\lambda_{t},Y_{t}) u_{Y}^{\star}(t)) + \nabla_{Z}b(t,Z_{t})^{\intercal} \mu(t,Z_{t})\\
& + \frac{1}{2}\ \tr\left( \sigma(t,Z_{t})^{\intercal} H_{Z}(b(t,Z_{t})) \sigma(t,Z_{t}) \right) \\
&+ \int_{\mathbb{R}^{k+1} \setminus \{0\}} b(t,Z_{t}+\Delta Z_{t}(x,\bar{x})) - b(t,Z_{t}) - (\Delta Z_{t}(x,\bar{x}))^{\intercal} \nabla_{Z}b(t,Z_{t}) \ \vartheta_{X,\bar{X}}(dx,d\bar{x}) =0,
\end{split}
\label{eq:pide}
\end{align}
with suitable boundary conditions for $a$ and $b$. Using separation of variables again, we find the ODE
\begin{align*}
\dot{a}(t) + a(t)r &=0, \\
a(T) &=1,
\end{align*}
leading to $a(t) = e^{r(T-t)}.$ Observe that $A(t)=a(t),$ so we can cancel some terms in the optimal strategies. 
Several further definitions are in order:

\begin{itemize}
	\item $\Theta_{1}^{\intercal} := \tilde{\mu}^{\intercal} (\tilde{\sigma}_{S} +\tilde{\rho}_{S} \xi)^{-1}$, i.e., $\Theta_{1}^{\intercal} \in \mathbb{R}^{m},$
	\item $\Theta_{2}(t,\lambda_{t},Y_{t}) := \frac{\nu_{L}(t,\lambda_{t},Y_{t})}{\sigma_{L}^{2}(t,\lambda_{t},Y_{t}) + \tilde{\eta}_{L}(t,\lambda_{t},Y_{t})},$
	\item $ C(t,\lambda_{t},Y_{t},x,\bar{x}) := \Theta_{1} \rho_{S} x + \Theta_{2}(t,\lambda_{t},Y_{t}) \eta_{L}(t,\lambda_{t},Y_{t},\bar{x}),$\\
	\item $\phi_{1_{i}}(t,Z_{t}) := \begin{cases} \Theta_{1} \sigma_{S} \sigma_{S^{i}} S_{t}^{i},\ &i=1,\dots,m,\\ \Theta_{2}(t,\lambda_{t},Y_{t}) \sigma_{\lambda}(t,\lambda_{t}) \sigma_{L}(t,\lambda_{t},Y_{t}) , \ &i=m+1, \\ \Theta_{2}(t,\lambda_{t},Y_{t}) \sigma_{L}^{2}(t,\lambda_{t},Y_{t})Y_{t}, \ &i=m+2, \end{cases}$\\
	i.e., $\phi_{1}(t,Z_{t}) = (\phi_{1_{1}}(t,Z_{t}),\dots,\phi_{1_{m+2}}(t,Z_{t}))\in \mathbb{R}^{m+2}.$\\
	%\item $\tilde{\Delta} Z_{t}(x,\bar{x}) := (\text{Diag}(S_{t-})\rho_{S}x + \Theta_{1}\rho_{S}x \textbf{1}, \tilde{\sigma}_{\lambda}(t,\lambda_{t-},\bar{x}) + \Theta_{2}(t,\lambda_{t},Y_{t}) \eta_{L}(t,\lambda_{t},Y_{t},\bar{x}), \\
	%Y_{t-} \eta_{L}(t,\lambda_{t-},Y_{t-},\bar{x}) + \Theta_{2}(t,\lambda_{t},Y_{t}) \eta_{L}(t,\lambda_{t},Y_{t},\bar{x}))^{\intercal},$ i.e., $\tilde{\Delta}Z_{t}(x,\bar{x}) \in \mathbb{R}^{m+2}.$\\
\end{itemize}
 Inserting $u_{S}^{\star}$ and $u_{Y}^{\star}$ into \eqref{eq:pide} and manipulating terms, we find that
\begingroup
\allowdisplaybreaks
\begin{align}
&\dot{b}(t,Z_{t}) + \frac{\Theta_{1}}{\gamma} \tilde{\mu} \ - \Theta_{1} \sigma_{S} \sum_{j=1}^{m} b_{z_{j}}(t,Z_{t}) S_{t}^{j} \sigma_{S^{j}} - \Theta_{1} \rho_{S} b_{1}(t,Z_{t}) + \frac{\Theta_{2}(t,\lambda_{t},Y_{t}) \nu_{L}(t,\lambda_{t},Y_{t})}{\gamma}\notag \ \\
&\ - \Theta_{2}(t,\lambda_{t},Y_{t}) \big(b_{z_{m+1}}(t,Z_{t}) \sigma_{\lambda}(t,\lambda_{t}) \sigma_{L}(t,\lambda_{t},Y_{t}) + b_{z_{m+2}}(t,Z_{t}) \sigma^{2}_{L}(t,\lambda_{t},Y_{t}) Y_{t}\big)\notag \ \\
&\ - \Theta_{2}(t,\lambda_{t},Y_{t}) b_{2}(t,Z_{t}) + \nabla_{Z}b(t,Z_{t})^{\intercal} \mu(t,Z_{t}) + \frac{1}{2}\ \tr\left( \sigma(t,Z_{t})^{\intercal} H_{Z}(b(t,Z_{t})) \sigma(t,Z_{t}) \right) \notag \ \\
&\ + \int_{\mathbb{R}^{k+1} \setminus \{0\}} b(t,Z_{t}+\Delta Z_{t}(x,\bar{x})) - b(t,Z_{t}) - (\Delta Z_{t}(x,\bar{x}))^{\intercal} \nabla_{Z}b(t,Z_{t}) \ \vartheta_{X,\bar{X}}(dx,d\bar{x}) =0 \notag \ \\
& \Rightarrow \dot{b}(t,Z_{t}) + \frac{\Theta_{1}}{\gamma} \tilde{\mu} + \frac{\Theta_{2}(t,\lambda_{t},Y_{t}) \nu_{L}(t,\lambda_{t},Y_{t})}{\gamma} - \Theta_{1} \sigma_{S}\sum_{j=1}^{m} b_{z_{j}}(t,Z_{t}) S_{t}^{j} \sigma_{S^{j}}\notag \ \\
&\ - \Theta_{2}(t,\lambda_{t},Y_{t}) \big(b_{z_{m+1}}(t,Z_{t}) \sigma_{\lambda}(t,\lambda_{t}) \sigma_{L}(t,\lambda_{t},Y_{t}) + b_{z_{m+2}}(t,Z_{t}) \sigma^{2}_{L}(t,\lambda_{t},Y_{t}) Y_{t}\big) \notag \ \\
&\ + \nabla_{Z}b(t,Z_{t})^{\intercal} \mu(t,Z_{t}) + \frac{1}{2}\ \tr\left( \sigma(t,Z_{t})^{\intercal} H_{Z}(b(t,Z_{t})) \sigma(t,Z_{t}) \right) \notag \ \\
&\ + \int_{\mathbb{R}^{k+1} \setminus \{0\}} \Big(\big(b(t,Z_{t}+\Delta Z_{t}(x,\bar{x})) - b(t,Z_{t})\big) (1-\Theta_{1} \rho_{S}x - \Theta_{2}(t,\lambda_{t},Y_{t}) \eta_{L}(t,\lambda_{t},Y_{t},\bar{x})) \notag \ \\
&\ - (\Delta Z_{t}(x,\bar{x}))^{\intercal} \nabla_{Z}b(t,Z_{t})\Big) \ \vartheta_{X,\bar{X}}(dx,d\bar{x}) =0 \notag \ \\
& \Rightarrow \dot{b}(t,Z_{t}) + \frac{\Theta_{1}}{\gamma} \tilde{\mu} + \frac{\Theta_{2}(t,\lambda_{t},Y_{t}) \nu_{L}(t,\lambda_{t},Y_{t})}{\gamma} - \Theta_{1} \sigma_{S} \sum_{j=1}^{m} b_{z_{j}}(t,Z_{t}) S_{t}^{j} \sigma_{S^{j}}\notag \ \\
&\ - \Theta_{2}(t,\lambda_{t},Y_{t}) \big(b_{z_{m+1}}(t,Z_{t}) \sigma_{\lambda}(t,\lambda_{t}) \sigma_{L}(t,\lambda_{t},Y_{t}) + b_{z_{m+2}}(t,Z_{t}) \sigma^{2}_{L}(t,\lambda_{t},Y_{t}) Y_{t}\big) \notag \ \\
&\ + \nabla_{Z}b(t,Z_{t})^{\intercal} \mu(t,Z_{t}) + \frac{1}{2}\ \tr\left( \sigma(t,Z_{t})^{\intercal} H_{Z}(b(t,Z_{t})) \sigma(t,Z_{t}) \right) \notag \ \\
&\ + \int_{\mathbb{R}^{k+1} \setminus \{0\}} \Big(\big(b(t,Z_{t}+\Delta Z_{t}(x,\bar{x})) - b(t,Z_{t})\notag \ \\ &\ - (\Delta Z_{t}(x,\bar{x}))^{\intercal} \nabla_{Z}b(t,Z_{t})\big) (1-C(t,\lambda_{t},Y_{t},x,\bar{x}))\Big)\ \vartheta_{X,\bar{X}}(dx,d\bar{x}) \notag \ \\
&\ - \int_{\mathbb{R}^{k+1} \setminus \{0\}} (\Delta Z_{t}(x,\bar{x}))^{\intercal} \nabla_{Z}b(t,Z_{t}) C(t,\lambda_{t},Y_{t},x,\bar{x})\ \vartheta_{X,\bar{X}}(dx,d\bar{x}) =0 \notag \\
\begin{split}
& \Rightarrow \dot{b}(t,Z_{t}) + \frac{\Theta_{1} \tilde{\mu} + \Theta_{2}(t,\lambda_{t},Y_{t})\nu_{L}(t,\lambda_{t},Y_{t})}{\gamma} + \nabla_{Z}b(t,Z_{t})^{\intercal} \Big(\mu(t,Z_{t}) - \phi_{1}(t,Z_{t}) \\ 
&\ - \int_{\mathbb{R}^{k+1} \setminus \{0\}}\Delta Z_{t}(x,\bar{x}) C(t,\lambda_{t},Y_{t},x,\bar{x})\ \vartheta_{X,\bar{X}}(dx,d\bar{x}) \Big)  + \frac{1}{2}\ \tr\left( \sigma(t,Z_{t})^{\intercal} H_{Z}(b(t,Z_{t})) \sigma(t,Z_{t}) \right) \\   
&\ + \int_{\mathbb{R}^{k+1} \setminus \{0\}} \Big(\big(b(t,Z_{t}+\Delta Z_{t}(x,\bar{x})) - b(t,Z_{t}) \\
&\ - (\Delta Z_{t}(x,\bar{x}))^{\intercal} \nabla_{Z}b(t,Z_{t})\big) (1-C(t,\lambda_{t},Y_{t},x,\bar{x}))\Big) \vartheta_{X,\bar{X}}(dx,d\bar{x}) = 0.
\label{eq:pide_b}
\end{split}
\end{align}
\endgroup

Observe that \eqref{eq:pide_b} is a linear PIDE and therefore solvable. In the sequel, we present a Feynman-Kac solution (\cite{fk}):

\begin{align}
\begin{split}
b(t,z) &= \mathbb{E}^{\mathbb{P}^{*}}_{t,z}\left[\int_{t}^{T}\frac{\Theta_{1}\tilde{\mu} + \Theta_{2}(s,\lambda^{*}_{s},Y^{*}_{s})\nu_{L}(s,\lambda^{*}_{s},Y^{*}_{s})}{\gamma}\ ds\right]\\
& = \frac{\Theta_{1}\tilde{\mu}}{\gamma}(T-t) + \mathbb{E}^{\mathbb{P}^{*}}_{t,z}\left[\int_{t}^{T} \frac{\Theta_{2}(s,\lambda^{*}_{s},Y^{*}_{s}) \nu_{L}(s,\lambda^{*}_{s},Y^{*}_{s})}{\gamma}\ ds\right].
\end{split}
\label{eq:b}
\end{align}

Note that this form of the function $b$ implies that the last two terms inside the brackets in formula for $u_{S}^{\star}(t)$ given by \eqref{eq:uS_prelim} are equal to zero. The dynamics of $Z^{*}$ under the measure $\mathbb{P}^{*}$ read 

\begin{align*}
dZ^{*}_{t} &= \left(\mu(t,Z^{*}_{t}) - \phi^{(1)}(t,Z_{t}^{*}) - \int_{\mathbb{R}^{k+1} \setminus \{0\}} \Delta Z_{t}^{*}(x,\bar{x}) C(t,\lambda_{t}^{*},Y_{t}^{*},x,\bar{x})\ \vartheta_{X,\bar{X}}(dx,d\bar{x}) \right) dt\\
&\ + \sigma(t,Z_{t}^{*})\ dW^{*}_{t} + \int_{\mathbb{R}^{k+1} \setminus \{0\}} \Delta Z_{t}^{*}(x,\bar{x})\ \tilde{J}^{*}_{X,\bar{X}}(dt,dx,d\bar{x}),
\end{align*}
with $Z_{0} = (S_{0},\lambda_{0},Y_{0})^{\intercal} \in \mathbb{R}^{m+2}.$
The density process 
$$\Phi_{t} = \frac{d\mathbb{P}}{d\mathbb{P}^{*}}\Bigg|_{\mathcal{F}_{t}}$$
solves the SDE
\begin{equation}
d\Phi_{t} = \Phi_{t} \psi^{(1)}(t,Z_{t})\ d\hat{W}_{t} + \Phi_{t-} \int_{\mathbb{R}^{k+1} \setminus \{0\}} \psi^{(2)}(t,Z_{t},x,\bar{x})\ \tilde{J}_{X,\bar{X}}(dt,dx,d\bar{x}),
\label{eq:dens_b}
\end{equation}
with $\Phi_{0} = 1.$ Moreover, $\psi^{(1)}(t,z)$ is given as solution of the system of equations\\
 $$\underbrace{\sigma(t,z)}_{(m+2)\times (d+1)} \cdot \underbrace{\psi^{(1)}(t,z)}_{(d+1)\times (1)} = \underbrace{-\phi^{(1)}(t,z)}_{(m+2) \times (1)}.$$
Note that there exists at least one solution to this system because $m \leq d$ (cf. the explanations preceding \eqref{eq:S}) and $\lambda$ and $Y$ are driven by the same Brownian motion $\bar{W}.$. In particular, it holds that $\psi^{(1)}_{d+1}(t,Z_{t}) = \frac{- \phi^{(1)}_{m+1}(t,Z_{t})}{\sigma_{\lambda}(t,\lambda_{t})} = \frac{- \phi^{(1)}_{m+2}(t,Z_{t})}{\sigma_{L}(t,\lambda_{t},Y_{t})Y_{t}}.$ Provided that $C(t,\lambda_{t},Y_{t},x,\bar{x}) < 1$ for Lebesgue-almost all $t \in [0,T]$ and $\vartheta_{X,\bar{X}}(dx,d\bar{x})$-a.s., it holds that
$$\psi_{2}(t,Z_{t},x,\bar{x}) = -C(t,\lambda_{t},Y_{t},x,\bar{x}).$$

If $\Phi$ is a positive martingale (see e.g. \cite{kazamaki}), then, according to the \textit{Girsanov theorem}, $\mathbb{P}^{*}$ is equivalent to $\mathbb{P}$ and 

\begin{align*}
d\hat{W}_{t} &= -\psi_{1}(t,Z_{t})\ dt + dW^{*}_{t}, \\
\vartheta^{*}_{X,\bar{X}}(dx,d\bar{x})&= (1-C(t,\lambda_{t},Y_{t},x,\bar{x}))\ \vartheta_{X,\bar{X}}(dx,d\bar{x}).
\end{align*}

Finally, we can represent the solution to the PIDE \eqref{eq:static} as
\begin{align}
\begin{split}
&B(t,z)\\
&:= \mathbb{E}_{t,z}\Bigg[\int_{t}^{T}\Big(e^{r(T-s)}(\tilde{\mu}^{\intercal}u_{S}^{\star}(s) + u_{Y}^{\star}(s)\nu_{L}(s,\lambda_{s},Y_{s}) \\
&\ \ - \frac{\gamma}{2} \left(u_{S}^{\star}(s)^{\intercal} \tilde{\sigma}_{S} u_{S}^{\star}(s) + u_{Y}^{\star}(s)^{2} \sigma_{L}^{2}(s,\lambda_{s},Y_{s}) \right) e^{2r(T-s)}\\
&\ \ - \gamma e^{r(T-s)} \nabla_{Z}b(s,z)^{\intercal} Q^{u^{\star}}(s,Z_{s})\\
&\ \ -\frac{\gamma}{2} \int_{\mathbb{R}^{k+1} \setminus \{0\}} (e^{r(T-s)}\Delta P_{s}^{u^{\star}}(x,\bar{x}) + b(s,Z_{s-}+\Delta Z_{s}(x,\bar{x})) - b(s,Z_{s-}))^{2}\ \vartheta_{X,\bar{X}}(dx,d\bar{x}) \Big) ds\Bigg].
\end{split}
\label{eq:sol_static}
\end{align}
Note that the expectation in \eqref{eq:sol_static} is calculated under the physical measure $\mathbb{P}.$ We summarize the most important part of the previous discussion in the following theorem:

\begin{theorem}
Consider the extended HJB system \eqref{eq:ext_hjb} for the case $D \equiv 0.$ For any $t \in [0,T)$ the optimal amounts to be invested in the stocks and the longevity asset are given by
\begin{equation}
u_{S}^{\star}(t) = \frac{\tilde{\mu} (\tilde{\sigma}_{S} +\tilde{\rho}_{S} \xi)^{-1}}{\gamma e^{r(T-t)}},
\end{equation}
\begin{align}
u_{Y}^{\star}(t) &= \frac{\nu_{L}(t,\lambda_{t},Y_{t}) -\gamma \big(b_{z_{m+1}}(t,z) \sigma_{\lambda}(t,\lambda_{t}) \sigma_{L}(t,\lambda_{t},Y_{t}) + b_{z_{m+2}}(t,z) \sigma^{2}_{L}(t,\lambda_{t},Y_{t}) Y_{t}\big) -\gamma b_{2}(t,z)}{\gamma (\sigma_{L}^{2}(t,\lambda_{t},Y_{t}) + \tilde{\eta}_{L}(t,\lambda_{t},Y_{t}))e^{r(T-t)}},
\end{align}
with the function $b$ given by \eqref{eq:b}, while $b_{2}$ is defined by \eqref{eq:b2}. In addition, the equilibrium value function $V$ decomposes into
$$V(t,p,z) = A(t)p + B(t,z),$$
with $A(t) = e^{r(T-t)}$ and  $B(t,z)$ given by \eqref{eq:sol_static}. The expected optimal terminal wealth $g_{u^{\star}}$ decomposes into
$$g_{u^{\star}}(t,p,z) = a(t)p + b(t,z),$$
with $a(t) = e^{r(T-t)}$ and $b(t,z)$ given by \eqref{eq:b}.
\label{thm_strategies}
\end{theorem}

\section{Numercal results}

In this part we exemplify Theorem \ref{thm_strategies}. The pricing of the longevity asset is done under some pricing measure $\mathbb{Q}$ while the optimization is performed under the objective measure $\mathbb{P}.$ Hence, we need to know the dynamics of $\lambda$ and $Y$ under both measures. Further, we need to choose a process modeling the force of mortality $\lambda$ that is nonnegative a.s. The Cox-Ingersoll-Ross (CIR) process supplemented by positive jumps (we refer to it as JCIR process in the sequel) is a good candidate for several reasons. The CIR process cannot become negative and belongs to the class of affine models allowing for a closed-form formula of the zero-bond price and the JCIR process preserves these properties (\cite{brigo}). We model the positive jumps by a compound Poisson process. Thereby the number of jumps is counted by the homogeneous Poisson process $\bar{N}=(\bar{N}_{t})_{t \in [0,T]}$ with constant intensity $\varrho_{\lambda} > 0$ and the jump sizes are independent and follow an exponential distribution with mean $\varsigma > 0.$  Thus, the L\'evy measure is given by $\vartheta_{\bar{X}}(d\bar{x}) = \varrho_{\lambda} f(d\bar{x}),$ with $f$ denoting the probability density function of an exponentially distributed random variable.
Let $\psi_{1}(t,Z_{t}) = \kappa \sqrt{\lambda_{t}}, \kappa > 0,$ be the market price of Brownian risk in the longevity market and denote by $\psi_{2} > -1$ the market price of jump risk. For parameters $\beta, \sigma_{\lambda}, \theta > 0$ and $\tilde{\theta} := \theta + \frac{(1+\psi_{2})\varrho_{\lambda} \varsigma}{\beta},$ consider the dynamics of $\lambda$ given by 
\begin{equation}
d\lambda_{t}= \big[\beta \tilde{\theta} - (\beta + \kappa \sigma_{\lambda}) \lambda_{t} - \psi_{2} \varrho_{\lambda} \varsigma \big]\ dt + \sigma_{\lambda} \sqrt{\lambda_{t}}\ d\bar{W}_{t} + \int_{\mathbb{R} \setminus \{0\}} \bar{x} \tilde{J}_{\bar{X}}(dt,d\bar{x}),
\label{eq:lambda_jcir_p}
\end{equation}
with $\lambda_{0} > 0.$ Defining
\begin{align}
\begin{split}
d\bar{W}_{t} &= dW^{\mathbb{Q}}_{t} + \kappa \sqrt{\lambda_{t}}\ dt, \\
\vartheta^{\mathbb{Q}}_{\bar{X}}(d\bar{x}) &= (1+\psi_{2})\ \vartheta_{\bar{X}}(d\bar{x}),
\end{split}
\label{eq:transformation_jcir}
\end{align}
a straightforward calculation shows that the $\mathbb{Q}$-dynamics of $\lambda$ reads
\begin{equation}
d\lambda_{t} = \beta[\theta - \lambda_{t}]\ dt + \sigma_{\lambda} \sqrt{\lambda_{t}}\ dW^{\mathbb{Q}}_{t} + \int_{\mathbb{R} \setminus \{0\}} \bar{x} \ J_{\bar{X}}(dt,d\bar{x}),
\label{eq:lambda_Q}
\end{equation}
$\lambda_{0} > 0$, which is the classical JCIR model under $\mathbb{Q}.$ From this we can easily deduce that an appropriate choice of parameters and the starting value $\lambda_{0}$ preserves nonnegativity. Note that we have specified the market prices of risk such that the model is tractable under both measures. In particular, the representation \eqref{eq:lambda_jcir_p} is consistent with the general setup in \eqref{eq:fom}, and \eqref{eq:lambda_Q} is the starting point for pricing. For $r \geq 0$ and $t \in [0,T),$ we are interested in the price of the zero-coupon longevity bond
$$L_{\lambda}(t,T) = \mathbb{E}_{\mathbb{Q}}\left[e^{-\int_{t}^{T}(\lambda_{s}+r)\ ds}\Big| \mathcal{F}_{t} \right].$$
Defining the auxiliary process
$$\tilde{L}_{\lambda}(t,T) := e^{r(T-t)} L_{\lambda}(t,T) = \mathbb{E}_{\mathbb{Q}}\left[e^{-\int_{t}^{T}\lambda_{s}\ ds}\Big| \mathcal{F}_{t} \right],$$
the affine structure of \eqref{eq:lambda_Q} yields that
\begin{equation}
\tilde{L}_{\lambda}(t,T) = A_{\lambda}(t,T) \alpha_{\lambda}(t,T) e^{-B_{\lambda}(t,T) \lambda_{t}},
\label{eq:tildeL}
\end{equation}
for deterministic functions $A_{\lambda}, \alpha_{\lambda},$ and $B_{\lambda}$ to be found in Chapter 22 of \cite{brigo}. Recall that the dollar value at time $t$ of an investment in $L_{\lambda}$ at time $t=0$ is given by $Y_{t} = e^{-\int_{0}^{t} \lambda_{s}\ ds} L_{\lambda}(t,T).$ The next proposition characterizes the dynamics of $Y$ assuming that $\lambda$ is modeled by a JCIR process. 

\begin{proposition}
Consider the process $\lambda$ given by \eqref{eq:lambda_jcir_p} and the specification \eqref{eq:transformation_jcir} of the pricing measure $\mathbb{Q}.$ The dollar value process $Y$ of an investment in the longevity asset $L_{\lambda}$ at time $t=0$ is given by 
\begin{align}
\begin{split}
\frac{dY_{t}}{Y_{t-}} &= \left(r + B_{\lambda}(t,T) \sigma_{\lambda} \kappa \lambda_{t} - \psi_{2} \int_{\mathbb{R} \setminus \{0\}} (e^{-B_{\lambda}(t,T) \bar{x}} -1) \vartheta_{\bar{X}}(d\bar{x}) \right) dt - B_{\lambda}(t,T) \sigma_{\lambda} \sqrt{\lambda_{t}}\ d\bar{W}_{t} \\
&\ + \int_{\mathbb{R} \setminus \{0\}} \left(e^{-B_{\lambda}(t,T)\bar{x}}-1 \right) \tilde{J}_{\bar{X}}(dt,d\bar{x}),
\end{split}
\label{eq:dY}
\end{align}
with $Y_{0} = L_{\lambda}(0,T).$
\end{proposition}
\begin{proof}
Using \eqref{eq:tildeL}, we define $f(t,\lambda_{t}) = \tilde{L}_{\lambda}(t,T).$ Then a standard calculation shows that
\begin{align}
\begin{split}
df(t,\lambda_{t}) &= \lambda_{t}f(t,\lambda_{t})\ dt + f'(t,\lambda_{t}) \sigma_{\lambda} \sqrt{\lambda_{t}} \ dW_{t}^{\mathbb{Q}} \\
&\ + \int_{\mathbb{R} \setminus \{0\}} \left(f(t,\lambda_{t-} + \bar{x}) - f(t,\lambda_{t-})\right) \tilde{J}^{\mathbb{Q}}_{\bar{X}}(dt,d\bar{x}).
\end{split}
\label{eq:df}
\end{align}
We easily see from \eqref{eq:tildeL} that
\begin{align*}
f'(t,\lambda_{t}) &= -B_{\lambda}(t,\lambda_{t}) f(t,\lambda_{t}), \\
f(t,\lambda_{t-} + \bar{x}) - f(t,\lambda_{t-}) &= f(t,\lambda_{t-}) \left(e^{-B_{\lambda}(t,T)\bar{x}}-1 \right).
\end{align*}
Plugging this back into \eqref{eq:df} and translating $df$ to the measure $\mathbb{P}$, we obtain
\begin{align*}
df(t,\lambda_{t}) &= f(t,\lambda_{t-}) \Bigg(\left(\lambda_{t} + B_{\lambda}(t,T) \sigma_{\lambda} \kappa \lambda_{t} - \psi_{2} \int_{\mathbb{R} \setminus \{0\}} (e^{-B_{\lambda}(t,T)\bar{x}}-1)\ \vartheta_{\bar{X}}(d\bar{x}) \right)\ dt \\
& \ - B_{\lambda}(t,T) \sigma_{\lambda} \sqrt{\lambda_{t}}\ d\bar{W}_{t} + \int_{\mathbb{R} \setminus \{0\}} (e^{-B_{\lambda}(t,T)\bar{x}}-1)\ \tilde{J}_{\bar{X}}(dt,d\bar{x}) \Bigg).
\end{align*}
Integration by parts then yields \eqref{eq:dY}. We conclude the proof by remarking that the function $B_{\lambda}$ is suitably integrable. 
\end{proof}
\noindent
We further deduce from \eqref{eq:dY} (cf. \eqref{eq:Y}) that 
\begin{align*}
\nu_{L}(t,\lambda_{t},Y_{t}) &= \nu_{L}(t,\lambda_{t}) = B_{\lambda}(t,T) \sigma_{\lambda} \kappa \lambda_{t} - \psi_{2} \int_{\mathbb{R} \setminus \{0\}} (e^{-B_{\lambda}(t,T) \bar{x}} -1)\ \vartheta_{\bar{X}}(d\bar{x}), \\
\sigma_{L}(t,\lambda_{t},Y_{t}) &= \sigma_{L}(t,\lambda_{t}) = - B_{\lambda}(t,T) \sigma_{\lambda} \sqrt{\lambda_{t}}, \\
\eta_{L}(t,\lambda_{t},Y_{t},\bar{x}) &= \eta_{L}(t,\bar{x}) = e^{-B_{\lambda}(t,T)\bar{x}}-1, \\
\tilde{\eta}_{L}(t,\lambda_{t},Y_{t}) &= \tilde{\eta}_{L}(t) = \int_{\mathbb{R} \setminus \{0\}} \eta_{L}(t,\bar{x})^{2}\ \vartheta_{\bar{X}}(d\bar{x}).
\end{align*}
In order to implement the strategies $u_{S}^{\star}$ and $u_{Y}^{\star},$ we need to specify the function $b(t,z) = b(t,\lambda)$ from \eqref{eq:b}. One can calculate that the density process \eqref{eq:dens_b} is given as solution of the SDE

\begin{equation}
\frac{d\Phi_{t}}{\Phi_{t-}} =  -\frac{\nu_{L}(t,\lambda_{t})}{\sigma_{L}^{2}(t,\lambda_{t}) + \tilde{\eta}_{L}(t)} \sigma_{L}(t,\lambda_{t})\ d\bar{W}_{t} - \frac{\nu_{L}(t,\lambda_{t})}{\sigma_{L}^{2}(t,\lambda_{t}) + \tilde{\eta}_{L}(t)}  \int_{\mathbb{R} \setminus \{0\}} \eta_{L}(t,\bar{x})\ \tilde{J}_{\bar{X}}(dt,d\bar{x}),
\label{eq:dens_jcir}
\end{equation}
$\Phi_{0}=1$, i.e., $\Phi_{t}$ is the stochastic exponential of the integrated right-hand side of \eqref{eq:dens_jcir}.\\ For simplicity we restrict to the case $d=k=1$ for the simulation, that is, there is only one stock traded on the market. We model the jumps of the stock price process by a homogeneous Poisson process $N = (N_{t})_{t \in [0,T]}$ with intensity $\varrho_{S} > 0.$ Denote the compensated version by $\tilde{N} = (\tilde{N}_{t})_{t \in [0,T]}.$ The dynamics of the stock price then reads

\begin{equation}
\frac{dS_{t}}{S_{t-}} = \mu \ dt + \sigma\ dW_{t} + \rho \ d\tilde{N}_{t},
\label{eq:SN}
\end{equation} 
with $S_{0} > 0$.
 Finally, we see that
\begin{align*}
\Theta_{1} &= \frac{\tilde{\mu}}{\sigma + \rho \xi}, \\
\Theta_{2}(t,\lambda_{t},Y_{t}) &= \Theta_{2}(t,\lambda_{t}) = \frac{\nu_{L}(t,\lambda_{t})}{\sigma_{L}^{2}(t,\lambda_{t}) + \tilde{\eta}_{L}(t)}.
\end{align*}
Therefore the function $b$ is given by
$$b(t,\lambda_{t}) = \frac{\Theta_{1}\tilde{\mu}}{\gamma}(T-t) + \mathbb{E}_{t,\lambda_{t}}\left[ \frac{\Phi_{T}}{\Phi_{t}}\int_{t}^{T} \frac{\Theta_{2}(s,\lambda_{s}) \nu_{L}(s,\lambda_{s})}{\gamma}\ ds\right],$$
and Theorem \ref{thm_strategies} implies that the optimal strategies in this market setup read
\begin{align}
u_{S}^{\star}(t) &= \frac{\tilde{\mu}}{(\sigma^{2} + \rho^{2} \varrho_{S})\ \gamma e^{r(T-t)}}, \label{eq:uSjcir}\\
u_{Y}^{\star}(t) &= \frac{\nu_{L}(t,\lambda_{t}) - \gamma b_{\lambda}(t,\lambda_{t}) \sigma_{\lambda} \sqrt{\lambda_{t}} \sigma_{L}(t,\lambda_{t}) - \gamma \int_{\mathbb{R} \setminus \{0\}}(b(t,\lambda_{t} + \bar{x}) - b(t,\lambda_{t})) \eta_{L}(t,\bar{x})\ \vartheta_{\bar{X}}(d\bar{x})}{(\sigma_{L}^{2}(t,\lambda_{t}) + \tilde{\eta}_{L}(t))\ \gamma e^{r(T-t)}}. \label{eq:uYjcir}
\end{align}

\begin{table}
\centering
\begin{tabular}{l|l|l|l|l|l|l|l|l|l|l|l|l|l|l|l|l}\hline
 $P_{0}$ &$T$ & $r$ & $\gamma$ & $S_{0}$& $\mu$ & $\sigma$ & $\rho$ & $\varrho_{S}$ & $\beta$ & $\sigma_{\lambda}$ & $\theta$ & $\kappa$ & $\psi_{2}$ & $\varrho_{\lambda}$ & $\lambda_{0}$ & $\varsigma$   \\ \hline
$1$ &$10$ & $.02$ & $2$ & $1$ & $.06$ & $.1$ & $.1$ & $3$ & .4 & .3 & .1 & .2& -.2& .5& .05& .001 \\ \hline
\end{tabular}
\caption{Parameter values}
\label{par_val}
\end{table}

\begin{table}
\centering
\begin{tabular}{c|c|c|c}\hline
  $\mathbb{E}[S_{T}]$ & $\text{Var}[S_{T}]$ & $\mathbb{E}[\lambda_{T}]$ & $\text{Var}[\lambda_{T}]$ \\  \hline
	$1.822$ & 1.633 & 0.096 & 0.0102 \\ \hline	
\end{tabular}
\caption{Expectation and Variance}
\label{E_Var}
\end{table}

%\begin{table}
%\centering
%\begin{tabular}{a|a|a|a|a|a}
%\multicolumn{2}{c}{$T=10$} & \multicolumn{2}{c}{$T = 15$} & \multicolumn{2}{c}{$T = 25$} \\ \hline
%$\mathbb{E}[P^{\star}_{T}]$ & $\text{Var}[P^{\star}_{T}]$ & $\mathbb{E}[P^{\star}_{T}]$ & $\text{Var}[P^{\star}_{T}]$ & $\mathbb{E}[P^{\star}_{T}]$ & $\text{Var}[P^{\star}_{T}]$ \\ \hline
%$1.4371$ & $0.1050$& $1.6728$ & $0.1616$ & $2.1852$ & $0.2689$ \\ \hline
%\end{tabular}
%\caption{Expectation and Variance with different horizons}
%\label{table_T}
%\end{table}

\begin{table}
\centering
\begin{tabular}{c c c c}
&$T=10$&$T=15$&$T=25$ \\ \hline
$\mathbb{E}[P^{\star}_{T}]$ & $1.4371$ & $1.6728$& $2.1852$  \\ \hline
ROER & $3.63\%$ & $3.43\%$& $3.13\%$  \\ \hline
 $\text{Var}[P^{\star}_{T}]$ &$0.1050$ & $0.1616$& $0.2689$ \\ \hline 	
\end{tabular}
\caption{Expectation and Variance with different horizons}
\label{table_T}
\end{table}

\begin{table}
\centering
\begin{tabular}{c c c c c c c}
&(A)&(B)&(C)&(D)&(E)&(F) \\ \hline
$\mathbb{E}[P^{\star}_{T}]$ & $1.4244$ & $1.4377$& $1.4399$& $1.4365$& $1.4379$& $1.4373$  \\ \hline
ROER & $3.54\%$ & $3.63\%$& $3.65\%$& $3.62\%$& $3.63\%$& $3.63\%$  \\ \hline
 $\text{Var}[P^{\star}_{T}]$ &$0.1009$ & $0.1128$& $0.1077$& $0.1079$& $0.109$& $0.1083$ \\ \hline 	
\end{tabular}
\caption{Expectation and Variance: (A): without longevity asset; (B): ignoring jumps; (C): Brownian risk only; (D): std. normally distributed jump sizes of $S$; (E): $T_{L} = 15$; (F): $T_{L} = 25$}
\label{E_Var_P_opt_S}
\end{table}

%\begin{table}
%\centering
%\begin{tabular}{a a|a a|a a|a a}
%\multicolumn{2}{c}{(A)} & \multicolumn{2}{c}{(B)} & \multicolumn{2}{c}{(C)} & \multicolumn{2}{c}{(D)} \\ \hline
  %$\mathbb{E}[P^{\star}_{T}]$ & $\text{Var}[P^{\star}_{T}]$ & $\mathbb{E}[P^{\star}_{T}]$ & $\text{Var}[P^{\star}_{T}]$ & $\mathbb{E}[P^{\star}_{T}]$ & $\text{Var}[P^{\star}_{T}]$ & $\mathbb{E}[P^{\star}_{T}]$ & $\text{Var}[P^{\star}_{T}]$ \\  \hline
	%$1.4222 $ & $0.1009$ & $1.4377$ & $0.1128$ & $1.4399$ & $0.1077$ & $1.4365$ & $0.1079$ \\ \hline	
%\end{tabular}
%\caption{Expectation and Variance of optimal portfolio, (A): without longevity asset; (B): ignoring jumps; (C): Brownian risk only; (D): std. normally distributed jump sizes of $S$}
%\label{E_Var_P_opt_S}
%\end{table}

%\begin{table}
%\centering
%\begin{tabular}{a|a|a|a}
%\multicolumn{2}{c}{$T_{L} = 15$} & \multicolumn{2}{c}{$T_{L} = 25$} \\ \hline
%$\mathbb{E}[P^{\star}_{T}]$ & $\text{Var}[P^{\star}_{T}]$ & $\mathbb{E}[P^{\star}_{T}]$ & $\text{Var}[P^{\star}_{T}]$ \\ \hline
%$1.4379$ & $0.109$ & $1.4373$ & $0.1083$ \\ \hline
%\end{tabular}
%\caption{Expectation and Variance with different longevity asset horizon}
%\label{table_TL}
%\end{table}

In Table \ref{par_val} the assigned parameter values are summarized. We chose values that are typical in the literature. This leads to the expectations and variances of respectively the stock price and the force of mortality given in Table \ref{E_Var}. These quantities have been calculated using the formulas given by Lemma \ref{mom_lam}. In order to compare the expected payoffs under different scenarios, we consider the \emph{rate of expected return} (ROER) given by ROER $= \ln(\mathbb{E}[P_{T}^{\star}])/T.$  Following the optimal strategies \eqref{eq:uSjcir} and \eqref{eq:uYjcir} yields the expectation, ROER and variance of the optimal terminal wealth displayed in the left panel of Table \ref{table_T}. Any modification made in the sequel will be compared against these values. The other two panels in Table \ref{table_T} show the expected value, ROER and variance of the terminal wealth when increasing the horizon to respectively $15$ and $25$ years. We see that the expected terminal payoffs and ROERs lie significantly above the final payoffs one would receive from investing in the riskless asset solely. The overall variance is of course slightly increasing when the investment horizon is prolonged, for that reason the insurance company is induced to invest less in the risky asset causing a minor decrease in the ROER over time. \\
In Table \ref{E_Var_P_opt_S} several further scenarios are analyzed. In Panel (A) the expectation, ROER and variance of the optimal terminal wealth without investing in the longevity asset are displayed. From \eqref{eq:uSjcir} we see that the amount invested in the stock is the same as in the previously described case while the money allocated to the longevity asset before is now invested in the riskless asset. Comparing the values in Panel (A) of Table \ref{E_Var_P_opt_S} to the left panel of Table \ref{table_T}, we see a slight decrease in the expected terminal wealth and in the ROER and a slight increase in the variance. Hence, the overall effect of not investing in the longevity asset is relatively small. Next, in Panel (B), we investigated the change in mean, ROER and variance of the optimal terminal wealth when keeping $\lambda$ and $S$ as in \eqref{eq:lambda_jcir_p} and  \eqref{eq:SN} respectively, but erroneously assuming that $\lambda$ and $S$ do not exhibit jumps. This corresponds to the scenario that an investor observes the variance of the stock and the force of mortality, but naively ascribes it to the Brownian components. Note that in this case the values of $u_{S}^{\star}$ do not change. We observe that the expected value ($1.4377$) and the ROER ($3.63\%)$ are not significantly different compared to the benchmark, while the variance increased to $0.1128$, which corresponds to a rise of $7.4 \%$. This slight increase in the variance stems from the fact that the hedging strategy does no longer account for the presence of jumps. However, the effect is relatively small, which is in line with our previous findings in Panel (A) of Table \ref{E_Var_P_opt_S} because the impact of the investment in the longevity asset under the allocation rule $u_{Y}^{\star}$ on the mean, ROER and variance of the optimal terminal wealth is relatively low in general.  
Hence, buying the longevity asset on average yields a higher payoff than the riskless investment and leads to some diversification. When hedging a terminal condition linked to the mortality rate in the insurance pool as discussed in Example \ref{ex_H}, the effect observed is likely to be stronger. In addition, our findings indicate that erroneously ignoring the jumps the force of mortality exhibits would then also lead to a significantly higher variance. \\
In Panel (C) of Table \ref{E_Var_P_opt_S} we investigated the effect of setting the jump intensity of $S$ and $\lambda$ to zero, so all uncertainty is stemming from Brownian risk. Thereby the expected values and variances of $S_{T}$ and $\lambda_{T}$ have been kept stable at the values depicted in Table \ref{E_Var} by adjusting the volatility parameters of the respective Brownian parts. However, the mean of the optimal terminal wealth just changed to $1.4399$, the ROER marginally increased to $3.65\%$, while the variance of the optimal terminal wealth slightly changed to $0.1077.$ The same phenomenon is shown in Panel (D): we replaced the constant jump of the Poisson process $N$ in the dynamics of $S$ by standard normally distributed jump sizes while keeping mean and variance of $S$ the same: we obtained a mean of $1.4365$, an ROER of $3.62\%$ and a variance of $0.1079$. Thus, we tentatively conclude that as long as we know mean and variance of the stock and the force of mortality, the expected optimal terminal wealth, the ROER and the variance are robust.\\
%Paths of the optimal investment amounts are displayed in Figure \ref{fig:strategies}. Both panels show that the optimal amount is smooth in time which is consistent with a result of \cite{wcw}. In Panel (B) of Figure \ref{fig:strategies} we see that the closer the maturity date, the optimal investment amount rapidly increases. A reasonable economic interpretation is that as the price of the zero-coupon longevity bond is approaching unity, the implied volatility of the price curve becomes small and therefore investments seem less risky from a mean-variance point of view.\\
 So far we assumed that the time to maturity of the longevity asset coincides with the insurance horizon. We investigated the change of mean, ROER and variance of the optimal terminal wealth when the time to maturity of the longevity asset, say $T_{L},$ is longer than the insurance horizon, which has been kept fixed at $T=10.$ The results are displayed in Panel (E) and Panel (F) of Table \ref{E_Var_P_opt_S}. We see that neither the expected value, the ROER nor the variance of the terminal optimal wealth change significantly. This result is highly important from a practical point of view because one cannot expect to find longevity assets whose times to maturity coincide with the insurance horizon, there are too few of them offered and traded. Our results show that picking an asset with a longer time to maturity does in particular not add to the variance of the terminal wealth. \\
Finally, Figure \ref{fig:opt_path} displays a path of the optimal portfolio process. The jumps of the underlying processes are clearly visible and the path shows a positive trend. In Figure \ref{fig:strategies}, paths of the optimal dollar amounts to be invested are plotted. From \eqref{eq:dY} we can see that the excess rate of return can become negative inducing the return of the longevity asset to be below the riskless rate. Thus, taking on a short position in $Y$ at several points in time is optimal and this can be clearly seen in Panel (B). Furthermore, as maturity is approaching, the optimal amount invested in the longevity asset is increasing because the function $B_{\lambda}(\cdot,T)$ in \eqref{eq:dY} is tending to zero as $t$ approaches $T$. The economic meaning for this investment behavior is that close to maturity, the price of the longevity bond converges to $1$ (cf. \eqref{eq:L_lambda}) and does not exhibit much variation anymore. Thus, the investment in the longevity bond becomes less risky.  \\
We conclude by remarking that a numerical analysis incorporating the hedging of a terminal condition should be conducted. In this case, closed-form solutions to the extended HJB system \eqref{eq:ext_hjb} cannot be obtained anymore, so one needs to resort to more involved numerical methods. We leave this direction for future research.

\begin{figure}
	\centering
		\includegraphics[width=0.60\textwidth]{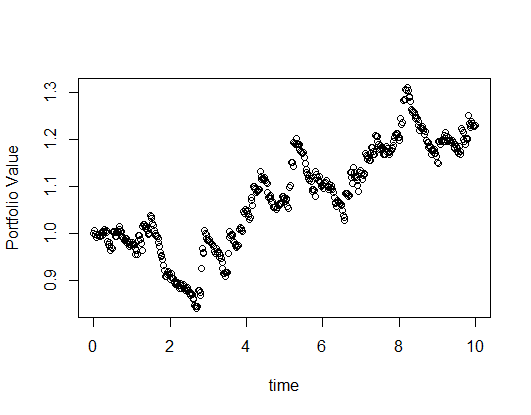}
		\caption{Optimal Portfolio Process}
		\label{fig:opt_path}
\end{figure}

%\begin{figure}
%\centering
%\begin{minipage}{.5\textwidth}
		%\centering
		%\includegraphics[width=0.50\textwidth]{opt_S.png}
		%%\caption{1a}
%\end{minipage}
%\begin{minipage}{.5\textwidth}
		%\centering
		%\includegraphics[width=0.50\textwidth]{opt_Y.png}
		%%\caption{2a}
%\end{minipage}
%\caption{Strategies}
%\end{figure}

\begin{figure}%
    \centering
    \subfloat[Stock]{{\includegraphics[width=7.5cm]{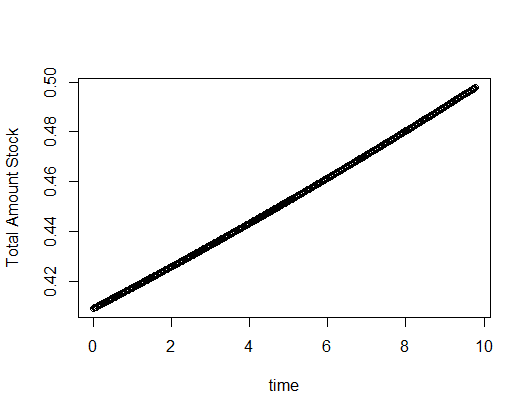} }}%
    \qquad
    \subfloat[Longevity Asset]{{\includegraphics[width=7.5cm]{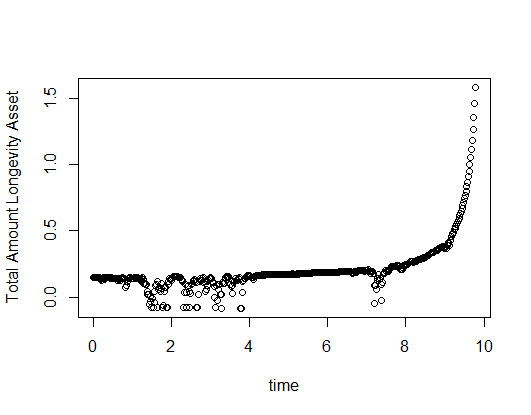} }}%
    \caption{Optimal Dollar Amounts}%
    \label{fig:strategies}%
\end{figure}

\newpage
%\section*{Appendix}
%\setcounter{section}{0}
%\setcounter{theorem}{0}
%\renewcommand{\thetheorem}{A\thechapter.\arabic{theorem}}
%\setcounter{equation}{0}
%\renewcommand{\theequation}{A\thechapter.\arabic{equation}}

\setcounter{section}{0}
\setcounter{theorem}{0}
\setcounter{equation}{0}
\renewcommand{\theequation}{\thesection.\arabic{equation}}

%\appendix
\section*{Appendix}
\renewcommand{\thesection}{A}

%\usepackage{chngcntr}
%\usepackage{apptools}
%\AtAppendix{\counterwithin{lemma}{section}}
%\appendix
%\section{}

\noindent 
We provide the formulas to calculate the moments of $\lambda_{T}$ displayed in Table \ref{E_Var}. 

\begin{lemma}
Consider the JCIR process given by \eqref{eq:lambda_jcir_p}, i.e., 
$$d\lambda_{t}= \big[\beta \tilde{\theta} - (\beta + \kappa \sigma_{\lambda}) \lambda_{t} - \psi_{2} \varrho_{\lambda} \varsigma \big]\ dt + \sigma_{\lambda} \sqrt{\lambda_{t}}\ d\bar{W}_{t} + \int_{\mathbb{R} \setminus \{0\}} \bar{x}\ \tilde{J}_{\bar{X}}(dt,d\bar{x}),$$
with $\lambda_{0} > 0.$
We have that
\begin{align}
\begin{split}
\mathbb{E}[\lambda_{t}] &= \frac{\beta \theta}{\beta + \kappa \sigma_{\lambda}}\left(1-e^{-(\beta + \kappa \sigma_{\lambda})T}\right) + \lambda_{0} e^{-(\beta + \kappa \sigma_{\lambda})T} + \frac{\varrho_{\lambda} \varsigma}{\beta + \kappa \sigma_{\lambda}}\left(1-e^{-(\beta + \kappa \sigma_{\lambda})T}\right), \\
\emph{Var}[\lambda_{t}] &= \frac{\beta \sigma_{\lambda}^{2}\theta}{2(\beta + \kappa \sigma_{\lambda})^{2}}\left(1-e^{-(\beta + \kappa \sigma_{\lambda})T}\right) + \frac{\lambda_{0} \sigma_{\lambda}^{2}}{\beta + \kappa \sigma_{\lambda}}e^{-(\beta + \kappa \sigma_{\lambda})T} \left(1-e^{-(\beta + \kappa \sigma_{\lambda})T}\right)\\
&\ \ + \frac{\varrho_{\lambda} \varsigma^{2}}{\beta + \kappa \sigma_{\lambda}} \left(1-e^{-2(\beta + \kappa \sigma_{\lambda})T} \right) + \frac{\sigma_{\lambda}^{2} \varrho_{\lambda} \varsigma}{\beta + \kappa \sigma_{\lambda}} \left(1-e^{-(\beta + \kappa \sigma_{\lambda})T} \right)\\
&\ \ + \frac{\sigma_{\lambda}^{2} \varrho_{\lambda} \varsigma}{2(\beta + \kappa \sigma_{\lambda})^{2}} \left(e^{-2(\beta + \kappa \sigma_{\lambda})T} -1 \right).
\end{split}
\label{eq:mom_lam}
\end{align}
\label{mom_lam}
\end{lemma}
\begin{proof}
Denote the imaginary unit by $\i.$ For notational simplicity, we define the following:
\begingroup
\allowdisplaybreaks
\begin{align*}
a &:= \beta + \kappa \sigma_{\lambda}, \\
b &:= \frac{\beta \theta}{a}, \\
g(t,y) &:= 1 + y\tilde{g}(t)\i,\\
\tilde{g}(t) &:= - \frac{\sigma_{\lambda}^{2}}{2a}\left(1-e^{-aT} \right), \\
\psi(t,y) &:= \frac{ye^{-at}\i}{g(t,y)}, \\
f_{1}(t,y) &:= g(t,y)^{-\frac{2ab}{\sigma_{\lambda}^{2}}}, \\
f_{2}(t,y) &:= \exp\left(\int_{0}^{t} \int_{0}^{\infty} \left(e^{\bar{x}\psi(s,y)}-1 \right)\ \vartheta_{\bar{X}}(d\bar{x})ds\right).
\end{align*}
\endgroup
According to \cite{rudiger}, the characteristic function of $\lambda_{t}$ reads
\begin{equation}
\mathbb{E}[e^{y \lambda_{t} \i}] = f_{1}(t,y)\ e^{\lambda_{0} \psi(t,y)}\ f_{2}(t,y).
\end{equation} 
So it holds that
\begin{align}
\frac{\partial}{\partial y} \mathbb{E}[e^{y \lambda_{t} \i}] &= f_{1}(t,y) e^{\lambda_{0} \psi(t,y)} f_{2}(t,y)\left[-\frac{2ab g'(t,y)}{\sigma_{\lambda}^{2}g(t,y)} + \frac{\lambda_{0}e^{-at}\i}{g(t,y)^{2}} + \int_{0}^{t} \int_{0}^{\infty} \frac{e^{\bar{x} \psi(s,y) - as}\bar{x} \i}{g(s,y)^{2}} \vartheta_{\bar{X}}(d\bar{x}) ds \right] \notag \\
&= \mathbb{E}[e^{y \lambda_{t} \i}] \cdot \left[\frac{b (1-e^{-at})\i}{g(t,y)} + \frac{\lambda_{0}e^{-at}\i}{g(t,y)^{2}} + \int_{0}^{t} \int_{0}^{\infty} \frac{e^{\bar{x} \psi(s,y) - as}\bar{x} \i}{g(s,y)^{2}} \vartheta_{\bar{X}}(d\bar{x}) ds \right] \label{eq:bracket},
\end{align}
and from this we easily deduce the first moment of $\lambda_{t}$:
\begin{align*}
\mathbb{E}[\lambda_{t}] &= (-\i)^{1} \cdot \frac{\partial}{\partial y} \mathbb{E}[e^{y \lambda_{t} \i}]\Bigg|_{y=0} \\
&=(-\i)^{1}\cdot 1\cdot \left[b (1-e^{-at})\i + \lambda_{0} e^{-at}\i + \i \frac{\varrho_{\lambda} \varsigma}{a} (1-e^{-at}) \right] \\
&= b (1-e^{-at}) + \lambda_{0} e^{-at} + \frac{\varrho_{\lambda} \varsigma}{a} (1-e^{-at}).
\end{align*}
Re-substitution of the abbreviations $a,b$ yields $\mathbb{E}[\lambda_{t}]$ given by \eqref{eq:mom_lam}. To calculate the second moment of $\lambda_{t}$, we need to determine the second derivative of the characteristic function w.r.t. $y$. The structure of the first derivative given by \eqref{eq:bracket} will turn out useful. For notational simplicity, we abbreviate the term in square brackets in \eqref{eq:bracket} by $[...]$ in the sequel. We find
\begin{align*}
\frac{\partial^{2}}{\partial y^{2}} \mathbb{E}[e^{y \lambda_{t} \i}] &= \frac{\partial}{\partial y}\left( \mathbb{E}[e^{y \lambda_{t} \i}] \cdot [...] \right) \\
&= \left(\frac{\partial}{\partial y} \mathbb{E}[e^{y \lambda_{t} \i}]\right) \cdot [...] + \mathbb{E}[e^{y \lambda_{t} \i}]\cdot \left(\frac{\partial}{\partial y}[...]\right) \\
&= \mathbb{E}[e^{y \lambda_{t} \i}] \cdot [...]^{2} + \mathbb{E}[e^{y \lambda_{t} \i}]\cdot \left(\frac{\partial}{\partial y}[...]\right) \\
&= \mathbb{E}[e^{y \lambda_{t} \i}] \left([...]^{2} + \frac{\partial}{\partial y}[...] \right).
\end{align*}
We need to calculate that

\begin{align*}
&\frac{\partial [...]}{\partial y} \\
&= \frac{b (1-e^{-at}) \tilde{g}(t)}{g(t,y)^{2}} + \frac{2\lambda_{0} e^{-at} \tilde{g}(t)}{g(t,y)^{3}}- \int_{0}^{t} \int_{0}^{\infty} \left(\bar{x}^{2} \frac{e^{\bar{x} \psi(s,y) - 2as}}{g(s,y)^{4}} - \frac{2\ e^{\bar{x} \psi(s,y) - as} \tilde{g}(s)}{g(s,y)^{3}} \bar{x}  \right) \vartheta_{\bar{X}}(d\bar{x}) ds,
\end{align*}
in order to argue that the second moment is given by

\begin{align*}
&\mathbb{E}[(\lambda_{t})^{2}] = (-\i)^{2} \cdot \frac{\partial^{2}}{\partial y^{2}} \mathbb{E}[e^{y \lambda_{t} \i}] \Bigg|_{y=0} \\
&= - \Big(-\mathbb{E}[\lambda_{t}]^{2} + b(1-e^{-at}) \tilde{g}(t)+2\lambda_{0} e^{-at} \tilde{g}(t) + \frac{\varsigma^{2} \varrho_{\lambda}}{a}(e^{-2at}-1)  + 2 \frac{\varrho_{\lambda} \varsigma}{a} \tilde{g}(t)\\
&\ \ - \frac{\sigma_{\lambda}^{2} \varrho_{\lambda} \varsigma}{2a^{2}}(e^{-2at}-1)\Big) \\
&= \mathbb{E}[\lambda_{t}]^{2} - b(1-e^{-at}) \tilde{g}(t)-2\lambda_{0} e^{-at} \tilde{g}(t) - \frac{\varsigma^{2} \varrho_{\lambda}}{a}(e^{-2at}-1)  - 2 \frac{\varrho_{\lambda} \varsigma}{a} \tilde{g}(t) + \frac{\sigma^{2} \varrho_{\lambda} \varsigma}{2a^{2}}(e^{-2at}-1).
\end{align*}
Finally, subtraction of the square of the first moment from the previous expression and re-substitution of the abbreviations $a,b,\tilde{g}(t)$ gives the formula for the variance of $\lambda_{t}$ provided by \eqref{eq:mom_lam}.
\end{proof}

\bibliographystyle{apalike}
\bibliography{bibliography}
\end{document}